\documentclass[a4paper,12pt]{article}
\usepackage[utf8]{inputenc}

\pdfoutput=1

\usepackage{setspace}
\usepackage[margin=1in,bmargin=.9in,tmargin=.9in]{geometry}
\usepackage{caption}
\usepackage[
  bookmarks=true,
  bookmarksnumbered=true,
  bookmarksopen=true,
  pdfborder={0 0 0},
  breaklinks=true,
  colorlinks=true,
  linkcolor=black,
  citecolor=black,
  filecolor=black,
  urlcolor=black,
]{hyperref}
\usepackage[round]{natbib}
\usepackage{bibunits}
\defaultbibliography{refs}
\defaultbibliographystyle{abbrvnat}
\usepackage{amssymb,amsmath,amsthm}
\usepackage{enumitem}
\usepackage[capitalise]{cleveref}
\Crefname{figure}{Figure}{Figures}
\usepackage{nicefrac}
\usepackage{isomath}
\usepackage{mathtools}
\usepackage[normalem]{ulem}

\usepackage[compact]{titlesec}
\usepackage[all=normal,paragraphs=tight,mathspacing=tight]{savetrees}
\setlist{itemsep=1pt,topsep=1pt,parsep=0pt}

\newcommand{\msg}[1]{{\tt (#1)}}
\newcommand{\func}[1]{{\tt #1}}

\theoremstyle{plain}
\newtheorem{theorem}{Theorem}[section]
\newtheorem{lemma}{Lemma}[section]
\theoremstyle{definition}
\newtheorem{definition}{Definition}[section]
\newtheorem{example}{Example}[section]

\usepackage{fnbreak}

\DeclareMathOperator{\poly}{poly}
\DeclareMathOperator{\polylog}{polylog}
\DeclareMathOperator{\subpoly}{subpoly}
\DeclareMathOperator*{\Exp}{\mathbb{E}}
\newcommand{\True}{\mathsf{True}}
\newcommand{\False}{\mathsf{False}}

\newcommand{\commitsubs}{c}
\newcommand{\commitrand}{\rho_\commitsubs}
\newcommand{\commitrandlen}{n_\commitsubs}
\newcommand{\commitmsg}{D_\commitsubs}
\newcommand{\commitmsglen}{L_\commitsubs}
\newcommand{\commitstrat}{S_\commitsubs}
\newcommand{\protcommitstrat}{\hat{S}_\commitsubs}
\newcommand{\commitverify}{\psi_\commitsubs}
\newcommand{\commitcorrectprop}{\phi_\commitsubs}

\newcommand{\runsubs}{r}
\newcommand{\runrand}{\rho_\runsubs}
\newcommand{\runrandlen}{n_\runsubs}
\newcommand{\runmsg}{D_\runsubs}
\newcommand{\runmsglen}{L_\runsubs}
\newcommand{\runstrat}{S_\runsubs}
\newcommand{\protrunstrat}{\hat{S}_\runsubs}
\newcommand{\runverify}{\psi_\runsubs}
\newcommand{\runcorrectprop}{\phi_\runsubs}

\newcommand{\simrand}{\rho_\Sim}
\newcommand{\simrandlen}{n_\Sim}

\newcommand{\ctprelation}{\eta}

\newcommand{\pprotocol}[5]{
{\begin{figure}[#4]
\begin{center}
\fbox{
    \footnotesize
   \hbox{\quad
   \begin{minipage}{5.25in}
  \begin{center}
    {\bf #1}
    \end{center}
    #5
    \end{minipage}
    \quad}
    }
\caption{\label{#3} #2}
\end{center}
\end{figure}
} }

\newcommand{\F}{{\cal F}}
\newcommand{\fctp}{\F_{\mbox{\sc ctp}}}
\newcommand{\fcom}{\F_{\mbox{\sc com}}}
\newcommand{\fnicom}{\F_{\mbox{\sc nicom}}}
\newcommand{\fzk}{\F_{\mbox{\sc zk}}}
\newcommand{\fnizk}{\F_{\mbox{\sc nizk}}}

\newcommand{\Sim}{{\mathcal S}}
\newcommand{\A}{{\mathcal A}}
\newcommand{\D}{{\mathcal D}}

\usepackage{fontawesome5}

\newcommand{\cryptofootnote}[1]{\begingroup\renewcommand{\thefootnote}{\faWrench\arabic{footnote}}\footnote{#1}\endgroup}
\let\cryptocref\cref

\hypersetup{
    pdftitle={Mechanism Design Without Disclosure: Committing to and Running Hidden Mechanisms},
    pdfauthor={Ran Canetti, Amos Fiat, Yannai A. Gonczarowski}
}

\title{Mechanism Design Without Disclosure:\\Committing to and Running Hidden Mechanisms\thanks{A one-page abstract of an earlier version of this paper, entitled ``Zero-Knowledge Mechanisms,'' appeared in the \textit{Proceedings of the 26th ACM Conference on Economics and Computation}. The authors thank Nina Bobkova, Ben Brooks, Eric Budish, Modibo Camara, Laura Doval,
Piotr Dworczak, George Georgiadis, Shafi Goldwasser, Sergiu Hart, Andreas Haupt, {\'A}ngel Hernando-Veciana, Zo{\"e} Hitzig, Emir Kamenica, Deniz Kattwinkel, Jacob Leshno, Shengwu Li, Eric Maskin, Steven Matthews, Stephen Morris, Roger Myerson, Noam Nisan, Mallesh Pai, David Parkes, Doron Ravid, Phil Reny, Assaf Romm, Joe Root, J{\'o}zsef S{\'a}kovics, Ran Shorrer, Ron Siegel, Tomasz Strzalecki, Alex Teytelboym, Clayton Thomas, Rakesh Vohra, and participants at the 2022 Cowles Conference on Economic Theory, 7th Lindau Meeting on Economic Sciences, INFORMS 2022, ASSA-NAWM 2023, the WZB Workshop on Strategy-Proofness and Beyond, INFORMS Workshop on Market Design 2024, and at seminars at The Hebrew University, Harvard, MIT, Penn State, VSET, NYU, UChicago, UPenn, A16z Crypto, UCL, LSE, and Oxford, for helpful comments and discussions. The authors gratefully acknowledge research support by the following sources. Canetti: Supported by DARPA under Agreement No.\ HR00112020021; any opinions, findings and conclusions or recommendations expressed in this material are those of the authors and do not necessarily reflect the views of the United States Government or DARPA\@. Gonczarowski: National Science Foundation (NSF-BSF grant No.\ 2343922) and Harvard FAS Dean's Competitive Fund for Promising Scholarship. Parts of the work of Gonczarowski were carried out while at Tel Aviv University and at Microsoft Research.
}}
\author{
Ran Canetti\thanks{Department of Computer Science, Boston University | \emph{E-mail}: \href{mailto:canetti@bu.edu}{canetti@bu.edu}.}
\and
Amos Fiat\thanks{Department of Computer Science, Tel Aviv University | \emph{E-mail}: \href{mailto:fiat@tau.ac.il}{fiat@tau.ac.il}.}
\and
Yannai A. Gonczarowski\thanks{Department of Economics and Department of Computer Science, Harvard University | \emph{E-mail}: \href{mailto:yannai@gonch.name}{yannai@gonch.name}.}
}
\def\documentdate{June 3, 2026}
\date{\documentdate}

\begin{document}

\maketitle

\begin{abstract}
A central tenet in mechanism design is the ability to irrevocably commit to a mechanism. Commitment is achieved by public declaration, letting players verify incentive properties in advance and the outcome in retrospect. However, public declaration can reveal superfluous information that is private to the mechanism designer, such as her target function or costs. We propose a new approach to commitment, and show how to commit to, and run, \emph{any given mechanism} without disclosing it, while enabling the verification of incentive properties and the outcome---all without any mediators. Our framework leverages zero-knowledge proofs---a cornerstone of modern cryptographic theory.
\end{abstract}

\pagenumbering{roman}
\thispagestyle{empty}

\clearpage

\setcounter{tocdepth}{2}
\begin{spacing}{1}
\tableofcontents
\end{spacing}

\thispagestyle{empty}
\clearpage
\pagenumbering{arabic}

\begin{bibunit}

{\raggedleft\small
\emph{``Don't tell customers more than they need to know.''}

\footnotesize ---Ferengi Rule of Acquisition \#39, Star Trek

\vspace{-1em}
}

\section{Introduction}

At the heart of the mechanism-design paradigm lies the designer's ability to commit to a mechanism.
Public declaration of the mechanism not only facilitates the commitment, but also serves to allow the players to inspect the mechanism and verify its strategic properties or other properties of interest. For example, a player's recognition of a dominant strategy as such enables the mechanism designer to predict play when designing the mechanism.

However, the mechanism being an ``open book'' reveals far more than incentive properties. By inspecting the mechanism, players might be able to infer information that the mechanism designer may have preferred to conceal. For example, consider a setting with one seller, one good for sale, and one buyer. Imagine that the seller posts the optimal monopoly price for the good---a price calculated based on both the seller's prior over the buyer's valuation and the seller's cost. Then, a buyer who knows the seller's prior could reverse the calculation and infer the seller's cost (under generic conditions). And, this cost might be a trade secret.\footnote{Why the seller might wish to keep her cost, or any other information, secret is unmodeled for now (and breaks away from the standard monopolist problem).} In more complex auction or selling settings, one could think not only of costs as being trade secrets, but also of other information that influences the design of the mechanism, such as inventory size, production capabilities, and even qualitative features such as whether the mechanism is constructed to maximize welfare or revenue.

One way to commit to a mechanism without disclosing any information except for its incentive properties is via a trusted mediator \citep[in the spirit of][]{Myerson1983}. In the monopoly-price example, the seller could discreetly entrust such a mediator with the 
price. The buyer would then declare her value (maximum willingness to pay), and the mediator would reveal the price if it is no more than the buyer's value, and otherwise only say that the price is higher than the value. With more general mechanisms, the mechanism designer could discreetly entrust the mediator with the mechanism description and instruct the mediator to confirm various properties of the mechanism (e.g., incentive properties such as individual rationality and incentive compatibility) to the players. After all bids are declared (or more generally, all actions are played), the mediator could confirm that the outcome declared by the designer indeed corresponds to the mechanism entrusted to her at the outset. The mediator would then never speak of this mechanism again.\footnote{Indeed, some trade secrets might be harmful if disclosed even years later.}

The availability of a trustworthy mediator, however, is a strong, often unrealistic, assumption. One compelling example comes from an episode involving the famous German poet Johann Wolfgang von Goethe, recounted by \cite{MoldovanuT1998}. 
Seeking to elicit his publisher's valuation of the right to publish one of Goethe's poems, Goethe made the publisher a BDM-style offer \citep*{BeckerDM1964}. Goethe trusted his personal lawyer to serve as a mediator and keep the reserve price in the offer secret. The lawyer, while trusted, turned out not to be trust\emph{worthy}: he leaked the price to the publisher, who could then bid precisely that price and no more. In their paper that recounts and analyzes this fascinating story, \cite{MoldovanuT1998} note:

\begin{quotation}\small
\noindent
``Commitments based on a secret reserve price are, in general, hard to achieve: if it is secret, what is to stop the seller from reneging? Part of Goethe's cleverness consisted in devising the scheme such that commitment is achieved by using a third, neutral party.''
\end{quotation}

In this paper, we investigate to what extent secrecy guarantees such as those provided when running a mechanism via a trustworthy mediator can be provided in the absence of a mediator, or in other words, to what extent public disclosure is an inseparable part of commitment to a mechanism when a mediator is not available.

The traditional \emph{protocol} that is used to commit (by public declaration) to an individually rational (IR) and incentive compatible (IC) direct-revelation mechanism\footnote{For simplicity, we focus in our exposition on IR and IC as the properties of interest of the mechanism that are to be verified before it is run. Our results apply to arbitrary properties of mechanisms.} and then run it can be formalized as consisting of three messages:
\begin{enumerate}
    \item \emph{Commitment:} The mechanism designer (e.g., seller) sends a message to the player (e.g., buyer),\footnote{For simplicity, in the introduction we focus on the case of one player. Our results apply to any number of players.} describing the mechanism. By virtue of this message, (a)~commitment is established (if the mechanism designer fails to follow the mechanism later then this will be publicly known and hence punishable), and (b) the player can verify that the mechanism is indeed IR and IC.
    \item \emph{Direct revelation:} The player decides whether or not to participate, and if so, sends a message to the mechanism designer, revealing the player's type.
    \item \emph{Running the mechanism:} The mechanism designer sends a message to the player, specifying the outcome. The player can verify that the outcome is consistent with the mechanism that was fixed in the first message.
\end{enumerate}

In contrast, with mechanism secrecy in mind, one may ideally be interested in a protocol along the lines of the following, too-good-to-be-true protocol:
\begin{enumerate}
    \item \emph{Commitment:} The mechanism designer sends a message (say, a sequence of numbers) to the player. By virtue of this message, (a) commitment is established (in the sense that a mechanism was fixed, and if the mechanism designer fails to follow the mechanism later then this will somehow be publicly known and hence punishable), and (b) the player can verify that the mechanism is indeed IR and IC. \textbf{However, this message must \emph{not} reveal anything about the mechanism except for it being IR and~IC.}
    \item \emph{Direct revelation:} The player decides whether or not to participate, and if so, sends a message to the mechanism designer, revealing the player's type.
    \item \emph{Running the mechanism:} The mechanism designer sends a message to the player, containing the outcome and additional information (say, a further sequence of numbers). By virtue of this additional information, the player can verify that the outcome is consistent with the mechanism that was fixed when the first message was sent. \textbf{However, it must be that nothing observed by the player can reveal anything about the mechanism that would not have been revealed had the mechanism been run via a trustworthy mediator, i.e., anything beyond the mechanism being IR and IC, its outcome for the specific player's type, and whatever can be inferred from the combination of these.}
\end{enumerate}
We emphasize that the desideratum is not for a limited commitment to parts of the mechanism while the designer can still change the rest, but for a \emph{full} commitment, fixing the entire mechanism, yet revealing only certain facets of it (IR, IC, and the outcome for the player's reported type).

Due to the inherent tension between commitment and nondisclosure (absent a trustworthy mediator), one might expect a tradeoff between the two.
Nonetheless, we show that \emph{practically}, no tradeoffs are needed here. That is, we show that a close variant of the above too-good-to-be-true protocol, which is only slightly weaker in a technical sense, can in fact be implemented. We construct a protocol that can be used to \textbf{commit to, and run, any mechanism, without disclosing it but while certifying its properties} (such as IR and IC, and/or any other properties that the mechanism designer wishes to expose, such as group strategyproofness, stability, etc.). Our protocol can be used to run mechanisms with \textbf{any number of players}, and players \textbf{retain the same strategy space as in the underlying direct-revelation mechanism}. The ``only'' difference between our protocol and the traditional protocol is the way the mechanism is described and outcomes are certified, which provides practically the same secrecy guarantee to the seller as running the mechanism via a trustworthy mediator does---which can be thought of as ``first-best privacy''---yet without the need for any mediators or any preexisting trust between any parties.

A main question is, of course, how to model our desiderata for such a protocol. We formulate two main, contrasting desiderata: \emph{hiding}---requiring that the player learns no more than in the ``first-best privacy'' mediated interaction---and \emph{committing}---requiring that if there is any violation by the mechanism designer (i.e., if the mechanism is not IR or IC, or if the mechanism designer does not follow the mechanism when calculating the outcome), then this becomes known to the player. Intuitively, the latter desideratum pushes toward revelation of the mechanism, contrasting with the former desideratum.

We start with \emph{hiding}: Quite a lot can potentially be learned about a mechanism from its outcome alone,\footnote{\cite{EilatEM2023} study the complementary problem of reducing what can be learned from the outcome.} let alone from both the outcome and the type reports that yield it.\footnote{Analysis showing that in a weighted Groves mechanism with no pivot rule, all weights can be completely recovered given a single outcome and the type reports that yield it appeared in an earlier version of this paper. We omit it here for brevity and to keep the exposition focused.} How might one model the idea that nothing more is learned about the mechanism beyond what is learned from it being IR and IC, its outcome, and the type reports that yield it? We formulate the \emph{hiding} desideratum by comparing two runs of our protocol that use distinct IR and IC mechanisms. We require that when the mechanisms are assigned to the two runs in random order, one's prior on which run corresponds to which mechanism is unchanged given the commitment-step messages sent in these two runs. Furthermore, we require that if the type reports and the outcome of the mechanism are the same in both runs, then given all messages sent in these two runs, one's prior on which run corresponds to which mechanism remains unchanged.\footnote{For randomized mechanisms, we require this prior to remain unchanged if the type reports and the \emph{realized} outcome of the mechanism are the same in both runs. That is, the distribution from which the outcome is drawn is hidden, and only the realized outcome is revealed.}

We move on to \emph{committing}: We model this desideratum by requiring that there exists a verification procedure on all messages sent while running the protocol---i.e., a procedure whose input is all messages sent and whose output is either ``pass'' or ``fail''---such that if the verification passes, then it must be that an IR and IC mechanism was fixed before the commitment-step message was sent, and the outcome declared in the running step is the outcome of this mechanism for the reported types.\footnote{For randomized mechanisms, we require that if the verification passes, then the declared outcome is furthermore faithfully drawn (i.e., drawn faithfully from the distribution specified by the mechanism).} Formally, we require that there exists a verification procedure, and additionally, for every mechanism-designer strategy there exists a mapping, from all the information known to the mechanism designer before sending the commitment-step message, to IR and IC mechanisms, such that whenever the verification passes, the declared outcome must be the outcome (or, for randomized mechanisms, a faithfully drawn outcome) of the mechanism given by the mapping. Of course, we also require that our protocol be \emph{implementing}: for every IR and IC mechanism, there is a mechanism-designer strategy that runs this mechanism and makes the verification pass.

These desiderata, in the ideal form just described, cannot be met simultaneously. Nevertheless, we construct a protocol that \emph{practically} satisfies them for all mechanisms. To model what we mean by ``practically,'' we draw inspiration from computer science, and specifically, from cryptography. You, our reader, are most likely reading this introduction within a web browser. This web browser, near the URL (website address) of the paper, might display a lock icon, which you trust to mean that any communication between your computer and the website is readable only by your computer and by the website. While this is what the icon \emph{practically} means, it is in fact not what it precisely means. What it precisely means is that (1) assuming that certain, well-studied computational problems (such as factoring large numbers) are computationally hard, and (2) assuming that any eavesdropping third party does not have exceedingly high computational power (say, more than all computers in the world combined), and (3) barring an exceedingly unlikely lucky guess by an eavesdropping third party (intuitively, correctly guessing a strong encryption password), any communication between your computer and the website is readable only by your computer and by the website (barring any bug in the browser or website software that implements the cryptography, of course).\footnote{\label{epsilon-edge}The guarantee with respect to point (3) is technically somewhat more nuanced, guaranteeing that the eavesdropper does not get any nonnegligible ``edge'' (whether by guessing a password or otherwise).}

Similarly, we henceforth use versions of our ``hiding'' and ``committing'' desiderata that require that their above versions hold (1) under standard, widely believed assumptions that certain, well-studied computational problems are computationally hard, and (2) assuming that the agents (the player for the hiding desideratum, the mechanism designer for the committing desideratum) do not have exceedingly high computational power, and (3) with probability $1\!-\!\varepsilon$ over some event that neither the mechanism designer nor the player can influence.\footnote{For hiding, our guarantee with respect to point (3) is again---as in the web-browser privacy guarantee (see \cref{epsilon-edge})---technically somewhat more nuanced, guaranteeing that the player does not get an ``edge'' of more than~$\varepsilon$ in terms of her ability to discern between possible mechanisms.}
Naturally, what is considered ``exceedingly high computational power,'' and what $\varepsilon$ is, are parameters of our construction.\footnote{Within computer science, cryptographic constructions are usually not directly parametrized by an attacker's computational power and success probability $\varepsilon$ as ours are, but rather by a parameter called a ``security parameter,'' which in turn implies both of these values but is harder to interpret from an economic perspective (which is the reason we avoid this parameter).}

Specifically, for hiding, we require that given the commitment-step messages sent in two runs corresponding to two distinct IR and IC mechanisms (assigned to the runs in random order), without using exceedingly high computational power one cannot correctly guess which run corresponds to which mechanism with probability greater than $\frac{1}{2}+\varepsilon$. Similarly, if the type reports and the (realized) outcome of the mechanism are the same in both runs, then given all messages sent in these two runs, without using exceedingly high computational power one cannot correctly guess which run corresponds to which mechanism with probability greater than $\frac{1}{2}+\varepsilon$. For committing, we require that for every mechanism-designer strategy that does not use exceedingly high computational power, there exists a mapping, from all of the information known to the mechanism designer before sending the commitment-step message, to IR and IC mechanisms, such that if the announced outcome is not the outcome of the mechanism (or, for randomized mechanisms, if it is not faithfully drawn), then with probability at least $1-\varepsilon$ the verification does not pass.

Finally, an important additional desideratum for our protocol is that it be \linebreak \emph{computationally unobtrusive}, a desideratum that rules out protocols satisfying our other desiderata at the cost of excessive computational friction for ``well-behaved'' agents (i.e., agents who participate in the protocol without malice). Specifically, the computational power required by well-behaved agents---for the mechanism designer, this is the computational power of the strategy from our \emph{implementing} desideratum; for the player, this is the computational power of running the verification (recall that our protocol is strategically equivalent to traditional ones, so the only message sent by the player is the single message that contains her type)---is essentially the same as in the traditional protocol, and smaller by many orders of magnitude than the computational power against which the protocol provides a guarantee (e.g., computational power just shy of that which is required to factor large numbers). Moreover, the computational power required by well-behaved agents increases exceedingly slowly as we decrease $\varepsilon$. In a precise sense, well-behaved agents that have moderate computational power are provided with very strong guarantees against adversaries with vastly greater computing power.

We prove that our four desiderata (in their practical form) can be simultaneously met: There is a protocol that is \emph{hiding}, \emph{committing}, \emph{implementing}, and \emph{computationally unobtrusive}. We emphasize that at no point during or after the run of the protocol is the mechanism ever disclosed to anyone by the designer. Rather than a traditional commitment to an observed mechanism via verification of publicly observed signals (the mechanism announcement and the outcome announcement), our protocol enables a \emph{self-policing} commitment to a \emph{never-observed mechanism}, yet (similarly to a traditional commitment) still via verification of publicly observed signals (the messages). And yet, as noted, this is achieved without relying on any mediators or preexisting trust and without changing the strategy space of any player, practically maintaining strategic equivalence to the traditional protocol in which the mechanism is publicly announced and then run. \textbf{In a precise sense, we decompose the classic notion of \emph{commitment to a mechanism}, showing that disclosure of the mechanism is, in fact, not an essential part thereof.} Our result can therefore also be viewed as an implementation-theory result: Any mechanism is implementable without disclosure.

Our existence proof for the protocol is constructive. While our ``committing'' and ``hiding'' desiderata (both their ``ideal'' and ``practical'' variants) differ substantively from standard definitions in computer science, to construct our protocol we leverage tools from cryptography. More precisely, we turn to zero-knowledge proofs, a cornerstone of modern cryptographic theory. At a \emph{very} high level, glossing over many details, the commitment-step message in our protocol contains a \emph{cryptographic commitment} to the mechanism (not to be confused with the game-theoretic notion of commitment, even though in this paper the former implements the latter). One might intuitively think of this cryptographic commitment as an encrypted version of the mechanism.\footnote{At the technical level, the requirements for commitments and encryptions are quite different.} The commitment-step message also contains a \emph{zero-knowledge proof} that the commitment is to a mechanism that is IR and IC---and/or satisfies any other properties that the mechanism designer wishes to expose. This zero-knowledge proof is a mathematical object that, when examined, convinces the player beyond any reasonable doubt (i.e., this is where $\varepsilon$ comes in) that the mechanism designer has created the cryptographic commitment by hiding a mechanism that satisfies the claimed properties, yet reveals no other information about the mechanism.\footnote{A main challenge to finding real-life applications for zero-knowledge proofs is that in many potential applications, it is unclear what \emph{precisely} the formal statement one wishes to prove is. (Such clarity is crucial since one learns nothing beyond the correctness of this specific formal statement.) Our paper identifies a domain in which there is wide consensus regarding the precise formal statement of interest (IR and IC).} The additional information in the running-step message then contains a \emph{zero-knowledge proof} that the declared outcome really is the outcome of the mechanism committed to in the commitment-step message, when applied to the player's reported type. Throughout this paper, footnotes that discuss noteworthy departures from standard cryptography, and which are aimed at readers with background in cryptography or who are interested in more technical discussions of contrasts with standard cryptography, are marked with a wrench sign.\cryptofootnote{See, e.g., the footnotes on our definition of committing and hiding protocols (\cref{committing,hiding}).}

Admittedly, there are many components involved, which might be overwhelming. For this reason, before we start defining our model, in \cref{examples} we present four illustrative examples that gradually demonstrate how everything fits together. For each example---ranging from a simple hidden price, through multiple hidden prices and proving incentive properties, to a randomized mechanism with only the realized outcome ever revealed---we consider a simple pricing or auction setting, and give a complete, self-contained construction of our protocol that does not rely on any cryptographic background, and delegates no part of the construction to external sources. To do this in a (relatively) simple, self-contained manner, we devise novel zero-knowledge proof techniques tailored to handle operations common in selling mechanisms with relative simplicity. In particular, we develop a simple way to commit to numbers (e.g., prices) and to give zero-knowledge proofs of properties such as inequalities between these numbers.\footnote{As we show, our novel construction easily generalizes also to arbitrarily complex committed information and proven properties.} The resulting novel protocols, built using simple tools, are (relatively) accessible, computationally very light, and can be easily implemented in practice.

After reviewing preliminaries in \cref{prelims}, in \cref{framework} we formulate our general framework and desiderata. In \cref{general-existence}, we present our general existence theorem. We emphasize that this theorem goes far beyond the four simple auction examples from \cref{examples}, and is completely general in terms of the settings and mechanisms hidden, with no functional-form restrictions. To give just two starkly different examples, our theorem can be used, on the one hand, to hide minute details of very complex mechanisms (e.g., hide the details of an intricate multidimensional randomized screening mechanism), and, on the other hand, to hide major qualitative features of mechanisms (e.g., in a matching setting, hide whether the mechanism that is run is Serial Dictatorship, Deferred Acceptance, or Top Trading Cycles). The proof of our general existence theorem is relegated to the appendix. Unlike our illustrative examples, this proof leverages state-of-the-art cryptographic tools in the construction of the protocol, and as such has to navigate the gap between our economics-centered desiderata and the computer-science cryptographic guarantees of these tools.\footnote{To make this proof (and the gaps that it navigates) as transparent as possible to readers without background in cryptography, in \cref{cp-sec} we restate the results from the cryptographic literature that we use in the proof by formulating a general, flexible ``one-stop-shop wrapper'' around all the existing cryptographic tools and definitions that are needed in our proof. In \cref{proofs}, in turn, we prove our general existence theorem utilizing the reformulated framework from \cref{cp-sec}. We hope that the self-contained nature and relatively lower barrier to entry of the reformulated, consolidated framework from \cref{cp-sec} might make it appealing also for independent use in other economic theory papers that may wish to avoid the learning curve associated with directly engaging with numerous papers from the cryptographic literature (as we do in \cryptocref{cp-proof-sec} in the supplementary material when we map out how the framework from \cref{cp-sec} follows from existing results).} We conclude in \cref{discussion}.

In \cref{contracts} in the supplementary material, we show that our framework applies also to contract-design settings with moral hazard (i.e., mechanisms with private actions rather than private types).\footnote{Our framework applies also to mechanisms with both private types and private actions.}\textsuperscript{,}\footnote{Additional extensions were introduced in an earlier version of this paper, including how to reduce the communication requirements \citep[in the sense of][]{NisanS2006,Segal2007} of our protocols to be on par with the communication requirement of running a commonly known direct-revelation mechanism, how players can commit to hidden strategies in sequential games while only proving certain properties of these strategies, and how to implement (without a mediator) a correlated equilibrium that is known \emph{only} to a principal and to none of the players, by proving only obedience and faithful randomization to the players. We omit these extensions here for brevity and to keep the exposition focused.} Using our framework, a principal can \textbf{commit to any contract} while proving arbitrary properties of the contract (such as IR, limited liability, and what the optimal level of effort to exert is), and eventually also prove that the contract was adhered to (i.e., that the wage paid is as specified by the contract for the publicly observable realized returns), all without revealing to the agent who agrees to the contract any further information about the contract (except what can be inferred from the combination of the proven properties and the realized returns and wage).\footnote{\cite{EdererHM2018} identify settings in which ex-ante committing to an ``opaque contract'' strictly dominates the best transparent contract. Our framework allows one to create such a commitment without any mediators or preexisting trust.}
For example, a common form of contract for a principal to offer to an agent is a fixed fraction of the net revenue (i.e., returns less the principal's cost). Using our framework, a principal can commit to follow such a contract that is hidden from the agent, while proving to the agent that the contract is of this functional form and incentivizes effort (i.e., that the wage expected by the agent if investing effort covers the agent's effort costs), and eventually prove that the contract was adhered to. The agent would eventually learn the overall returns and her wage, but would learn neither the precise fraction of the revenue that she received nor the principal's precise cost (she would only learn some joint function of these two values). These details of the contract that the agent accepted would remain forever hidden from the agent. Among other uses, hiding contracts allows a principal to hire several agents without allowing them to compare the detailed terms of their contracts.\footnote{\cite{CullenP2022} find that an increase in employees' perception of their peers' salaries has a negative causal effect on effort and performance.} 

While our framework is theoretically sound for use by ``Homo economicus'' agents and AI agents, one may rightfully wonder about the applicability of our framework for real-life use by humans. To answer this question, we once again recall a success story of introducing equally technical mathematics into mass use: that of web-browser communication encryption, symbolized by the browser's lock icon. Users who surf the World Wide Web are spared from directly interacting with the intricate mathematics required for encrypting their browser communication, since the browser software handles all of the mathematics. The browser software itself is trusted by users because it is written by a reputable party and audited by many independent experts. Note, however, that this reputable party is not a mediator in the traditional sense: it is never entrusted with the browser communication; this communication only exists unencrypted on the user's computer and on the website server. This model has worked successfully for three decades, and a very similar model could be used for implementing the machinery in this paper: A reputable party could write the software for the mechanism designer and players, which would be audited by independent experts.\footnote{This could possibly even be part of the browser software in the case of online auctions, where an icon similar to the browser's lock icon---perhaps a stylized check mark (leveraging its older sibling icon's decades-old credibility)---would confirm the validity of claims made by the designer.} The mechanism description would exist unencrypted only on the mechanism designer's computer, yet its properties and outcome would be trusted by all users.\footnote{In this setting, our results guarantee that the mechanism description is 
kept secret even if players use any other software, and the mechanism properties and outcome can be relied on even if the mechanism designer uses any other software. Each user need only trust that the software she is using is coded properly (as one trusts any other software on one's computer).}

\subsection{Related Literature}

Several papers, mostly within computer science, aim to apply cryptography to economics. Closest to our paper is the literature within mechanism design that focuses on the privacy of player types/valuations (for auctions, see, e.g., \citealp{Brandt2006,ParkesRST2008,ParkesRT2009,MicaliR2014,FerreiraW2020}; and see \citealp{BogetoftCDGJKNNNPST09} for a real-life application to auctioning allowances for growing sugar beets in Denmark; for matching mechanisms, see, e.g., \citealp{DoernerES2016}; for more general mechanisms, see \citealp{IzmalkovLM2011}; see \citealp{HauptH2022} for a non-cryptographic approach). In these papers, in contrast to ours, the mechanism is \emph{commonly known}, and the challenge is to run this commonly known mechanism without revealing the types of the players. Our focus is different, conceptually and technically: We focus on the secrecy of \emph{the mechanism itself}, which introduces new challenges both on the modeling/desiderata side---e.g., modeling commitment without disclosure---and on the technical side---e.g., proving/certifying certain properties of the hidden information (e.g., IR and IC; or IR, limited liability, and the optimal level of effort to exert). In the age of social networks, it is indeed not far-fetched to have settings in which companies value their privacy considerably more than users do. 

Another feature of this previous literature on cryptography and economics is that hiding the players' types without a mediator or physical means involves changing the extensive game form, enlarging the strategy spaces of the players, which creates theoretical room for various threats, considerably degrading strategic properties. In contrast, our framework allows for clear separation between the cryptographic and game-theoretic layers: Our protocol---despite hiding the mechanism---is a direct-revelation protocol, maintaining the same strategy space for each player as in the traditional protocol.\footnote{\label{hide-types}An extension of our framework that combines our approach with that of this previous literature, allowing both the mechanism and type reports to be hidden (with the same caveats as in this previous literature), was described in an earlier version of this paper. We omit it here for brevity and to keep the exposition focused.}

Also at the interface of cryptography and economics is the growing literature on blockchain protocols such as Bitcoin and Ethereum, reviewing even a small portion thereof is outside the scope of this paper. While one could certainly deploy the machinery that we develop in this paper on top of a blockchain (e.g., enrich smart contracts to implement our protocol), running a blockchain involves a complex infrastructure and monetary reward model that give rise to intricate incentives, which is not at all required by our machinery. Furthermore, in stark contrast with running a blockchain, the computational resources required by all of the illustrative examples that we give in \cref{examples} are very low. (So low, in fact, that they could be directly run even on the buyer's and seller's smartwatches.)

\section{Illustrative Examples}\label{examples}

In this \lcnamecref{examples}, we present a sequence of four examples, all within fundamental auction settings, each gradually building on its predecessor and illustrating an additional key property of our framework. In each example, we give a complete, self-contained construction. To do so, we devise novel zero-knowledge proof techniques tailored to handle operations common in selling mechanisms with relative simplicity. The resulting novel protocols are not only (relatively) accessible, but also computationally very light and easily implementable in practice.

The first two examples demonstrate the idea of a commitment to a hidden mechanism and proving that it is properly run. 
The third example adds an additional layer: proving that the hidden committed mechanism satisfies any property of interest (in that example, incentive compatibility). 
Finally, the fourth example deals with randomized mechanisms---i.e., mechanisms that map buyer reports to \emph{distributions} over outcomes---and shows how to ensure that the buyer learns nothing about the mechanism \emph{or even the distribution}, except for the \emph{realized} outcome---a single draw from this distribution.

While the constructions described in this \lcnamecref{examples} provably satisfy strong desiderata, our goal in this \lcnamecref{examples} is to give sufficient information and intuition for the reader to see how the claims of the previous paragraph might \emph{plausibly} be achieved, without fully proving this or formalizing what we mean by these claims. Our desiderata are formulated in \cref{framework}, and our theorem on the existence of a general construction is presented in \cref{general-existence}.
For completeness, in \cref{app:examples} in the supplementary material we prove in full detail that the desiderata from \cref{framework} are satisfied by the specific construction given in the first two examples in the current \lcnamecref{examples}; the proofs for the latter two examples are similar.

\subsection[Illustrative Example 1:\texorpdfstring{\\}{ }Commitment to Hidden Mechanism \&\ Proof of Properly Running It]{Illustrative Example 1: Commitment to Hidden Mechanism \& Proof of Properly Running It}

We start with the simplest possible example within an auction setting: a seller offering an item to a buyer for some price.\footnote{We avoid the term ``posted price'' since the price is secret and therefore not posted anywhere.} The buyer has a private value $v\in\{0,\ldots,H\!-\!1\}$ for the item, for some commonly known $H$, and the seller wishes to commit to a price $s\in\{0,\ldots,H\!-\!1\}$. For simplicity, we assume that $H$ is a power of $2$. The challenge is that the price should remain hidden if trade does not take place. More specifically, the buyer should learn nothing about $s$ before disclosing $v$, and if trade does not take place, then the buyer should learn that $s>v$ but learn no more than that. The buyer should, of course, be convinced that there really is a commitment here, i.e., that the seller has no way to alter the (hidden) price after she has set it. We describe the process by which the seller commits to the price, and then describe the process by which she proves that the outcome that she announces (i.e., whether trade takes place and at which price) correctly corresponds to the hidden price. The correctness of these processes relies on an algebraic construction that we also use in later examples in this \lcnamecref{examples}, for which we now give some background.\footnote{While this algebraic setup might at first glance make our examples seem impractical, we note that the same could have been equally (if not more so) said in the 1990s about the cryptography that powers the web browser with which you securely accessed this paper. As discussed in the introduction, just as the latter entered mainstream use, we conjecture that so could our techniques.}

\paragraph{Algebraic setup}
Let $p$ be a large prime such that $q=2p+1$ is also prime (e.g., $p=3$ and $q=7$, though we will use much larger primes).\footnote{Such a prime $p$ is called a \emph{Sophie Germain prime} and such a prime $q$ is called a \emph{safe prime}. Since we are interested in constructing a large group---whose elements are infeasible to enumerate one by one---a better example than $p=3$ and $q=7$ would be, e.g., $p=2618163402417\cdot2^{1290000} - 1$ and $q=2618163402417\cdot2^{1290001} - 1$, which, as of May 2026, are the largest known primes with this property (though infinitely many are conjectured to exist). See \url{https://primes.utm.edu/primes/page.php?id=121330} and \url{https://primes.utm.edu/primes/page.php?id=121331}.}
Let $G=\bigl\{h^2\!\bmod q\mid h\in\{1,\ldots,q-1\}\bigr\}$ be the set of all square integers modulo $q$. For every $g\in G\setminus\{1\}$, it is well known that $g^p\bmod q=1$ and that this is not true for any smaller positive power of $g$.\footnote{The first part is true because by construction $g=h^2\bmod q$ for some $h$, and by Fermat's little theorem, $g^p\bmod q=h^{2p}\bmod q=h^{q-1}\bmod q=1$ because $q$ is prime. For the second part, let $\ell>0$ be the smallest positive integer such that $g^\ell\bmod q=1$. $\ell$ must divide $p$ since otherwise the greatest common divisor of $\ell$ and $p$ would be an even smaller such integer. Since $p$ is prime, $\ell=p$.}
Henceforth, all operations involving elements of $G$ are modulo $q$.\footnote{The set $G$, which we use for simplicity of exposition, is a well-known special case of an \emph{(algebraic) group} of some large prime \emph{order} $p$.} 
Since for every $g\in G\setminus\{1\}$, the smallest positive power of $g$ that equals $1$ is $p$, any such $g$ is a \emph{generator} of~$G$, i.e., for every such $g$ and every $h\in G$ there exists a unique $\ell\in\{1,\ldots,p\}$ such that $h=g^\ell$. Hence, for every $g,h\in G\setminus\{1\}$ there exists a unique $\ell\in\{1,\ldots,p\!-\!1\}$ such that $h=g^\ell$. This $\ell$ is called the \emph{discrete logarithm (in $G$) of $h$ base~$g$}. It is widely believed by cryptographers and computational complexity theorists that if $G$ is large, then for uniformly random $g,h\in G\setminus\{1\}$ it is a hard problem to calculate the discrete logarithm of $h$ base~$g$, that is, it is (computationally) \emph{infeasible} to do so except with negligible probability (corresponding, e.g., to an extremely lucky very-low-probability guess of the discrete logarithm).\footnote{For simplicity, the protocols we present in this section are under the assumption that no agent (seller or buyer) has access to a quantum computer. Should quantum computers become a reality, then these protocols would need to be modified, as such computers are expected to be able to feasibly calculate discrete logarithms \citep{Shor97}. There are several alternative ``quantum safe'' hard problems on which cryptography can be based \citep{postquantum,NISTpostquantum}. In practice, such alternatives have not replaced standard cryptographic algorithms in many major applications.} The committing and hiding properties in this pricing example will follow from this infeasibility: Were anyone able to circumvent these properties, they would also be immediately able to feasibly solve this problem, disproving its hardness. We note that the security of many standard current cryptographic schemes for protecting internet communication hinges on this infeasibility. That is, anyone able to feasibly solve this computational problem would be able to break into many real-life cryptographic systems.\footnote{While we present our protocols in this \lcnamecref{examples} for the group $G$ for simplicity, these protocols work with any group of prime order where computing discrete logarithms is hard, such as elliptic curve groups, where the hardness of this problem is the basis for major security standards.}

\paragraph{Commitment scheme}
The commitment process is as follows: nature draws two elements $g,h\in G\setminus\{1\}$ uniformly at random, from which point $g$ and $h$ are commonly known.
To commit to a price $s\in\{0,\ldots,H\!-\!1\}$, write $s$ in binary representation: $s=s_1\ldots s_{\log_2H}$, i.e., $s=\sum_{i=1}^{\log_2H}s_i\cdot2^{\log_2H-i}$. For each $i\in\{1,\ldots,\log_2H\}$, the seller independently draws $\rho_i \sim U\{1,\ldots,p-1\}$, and calculates the \emph{commitment}~$C_i\in G\setminus\{1\}$ to the bit $s_i$ as $C_i=g^{\rho_i}$ if $s_i=0$, and as $C_i=h^{\rho_i}$ if $s_i=1$. (Recall that these powers are modulo $q$.)
The seller then sends $(C_1,\ldots,C_{\log_2H})$ to the buyer as a \emph{commitment} to the (unknown to the buyer) price~$s$.
Why do we consider this a hiding commitment? First, regarding hiding: For each~$i$, when $\rho_i$ is drawn uniformly at random, for each $f\in\{g,h\}$ the value $f^{\rho_i}$ is a uniformly random element of $G\setminus\{1\}$, and hence without knowing $\rho_i$, the value $f^{\rho_i}$ leaks no information about whether $f=g$ or $f=h$, and hence no information about $s_i$. Second, regarding commitment: This hinges on the fact that (with high probability over the draw of $g$ and $h$) the seller is unable to come up with $C_i,\rho_i,\rho_i'$ such that $g^{\rho_i}=C_i$ and at the same time $h^{\rho_i'}=C_i$.
Indeed, if the seller could have come up with such $C_i,\rho_i,\rho_i'$ then she could have feasibly calculated $\ell=\rho_i\cdot {\rho'_i}^{-1}\bmod p$,\footnote{Unlike discrete logarithms, the inverse ${\rho'_i}^{-1}\bmod\;p$ (which is well defined since $\rho'_i\ne 0\bmod\;p$) is easy to feasibly calculate: by Fermat's little theorem, ${\rho'_i}^{p-1}\bmod\;p=1$ because $p$ is prime, and therefore, ${\rho'_i}^{-1}\bmod\;p={\rho'_i}^{p-2}\bmod\;p$.} which satisfies $h=g^\ell$, and hence such a seller would have succeeded in feasibly calculating the discrete logarithm of $h$ base~$g$, which, by assumption, is computationally infeasible. Therefore, if $s_i=0$ then the seller knows the discrete log of $C_i$ base $g$ but not base $h$, and otherwise knows the discrete log of $C_i$ base $h$ but not base $g$. As we will see below, this prevents the seller from acting according to any price that does not equal $s$.

\paragraph{Proof of properly running the hidden mechanism} After the seller sends the commitment $(C_1,\ldots,C_{\log_2H})$ to the buyer, the buyer discloses her value $v$. If the price $s$ is no more than $v$, then the seller discloses $s$ along with the values $\rho_i$, which allow the buyer to verify the value of $s$ (recall that the seller only knows how to ``open'' each bit this way to $0$ or to $1$ but not to both), and trade takes place at price $s$. On the other hand, if the price $s$ is more than~$v$ then the seller must prove to the buyer that this is the case, without disclosing any additional information about $s$. We now describe how this is done.

Before we continue, consider for a moment two numbers in binary representation, $x=x_1\ldots x_{\log_2H}$ and $y=y_1\ldots y_{\log_2H}$.
Observe that $x<y$ if and only if there exists an $i$ such that both (1) $x_i<y_i$ and (2) $x_j\leq y_j$ for all $j<i$. Rephrasing this statement, $x\geq y$ if and only if for all $i\in\{1,\ldots,\log_2H\}$, either $x_i\geq y_i$, or else there exists some $j<i$ such that $x_j>y_j$. Equivalently, for all $i\in\{1,\ldots,\log_2H\}$ such that $y_i=1$, either $x_i=1$ or else $x_j=1$ for some $j<i$ such that $y_j=0$.

Back to our seller and buyer, let $v_1\ldots v_{\log_2H}$ be the binary representation of~$v\!+\!1$ (assuming $v < H-1$; otherwise, trivially $v\ge s$, trade takes place, and no proof is needed).
To prove that $s>v$, or equivalently that $s\geq v\!+\!1$ (where $v$ is commonly known and~$s$ is given implicitly via the commitment $(C_1,\ldots,C_{\log_2H})$), by the previous paragraph it suffices for the seller to prove, for each $i$ such that $v_i=1$, that either $s_i=1$ or else $s_j=1$ for some $j<i$ such that $v_j=0$. That is, for each such $i$ it suffices for the seller to prove that she knows the discrete logarithm base~$h$ of at least one of $\bigl\{C_j~\big|~j=i \vee (j<i \wedge v_j=0)\bigr\}$. Fortunately, a procedure for an agent to prove that she knows a discrete logarithm with a given base of one of a set of numbers, without revealing anything about for which of those numbers she knows this, is well known \citep{CDS94}. Specifically, there is a procedure that has the following three properties. First, there is only negligible probability that a seller who does not know the discrete logarithm base $h$ of any number in the set does not get discovered during this procedure. Second, the procedure reveals nothing beyond the seller knowing the discrete logarithm base $h$ of one of the numbers in the set. Third, participating in the procedure requires only modest computational power. For completeness, this procedure is given in \cref{cds} in the supplementary material.

\paragraph{On interaction} The main difference between this example (as well as the other examples in this \lcnamecref{examples}) and what our general framework in \cref{framework,general-existence} guarantees is that in this example, the final proof is interactive, i.e., the procedure for proving knowledge of a discrete log actively involves the buyer, and in particular involves the buyer sending a random message (drawn by the buyer) to the seller. While adding interaction might in general introduce game-theoretic complications as it enlarges the strategy space and allows for various threats, this is not a problem here, even if there are multiple buyers, as all buyers finish disclosing their types already before this point, so no threats are possible. Nonetheless, in our general construction any interaction here is avoided, and the proof of properly running the hidden mechanism consists of a single message sent by the seller. Our general framework provably achieves this by leveraging a more elaborate cryptographic scheme.\footnote{One way in which this non-interactivity on the buyer's side could be achieved at least in practice using the above cryptographic scheme, in this example and those in the remainder of this \lcnamecref{examples},
is to replace the random message sent by the buyer with bits obtained by applying a cryptographic hash function to all previous seller messages; as these bits are not known before those messages are constructed, they are in some sense ``random enough from the point of view of when messages are constructed.'' This technique, known as the \emph{Fiat--Shamir transform} \citep{FS86}, is used by countless real-life cryptographic systems.}

\subsubsection{Multiple Buyers}

In the above example we analyze a single-buyer setting for simplicity; however, everything described, with no material change, can also be used to implement a second-price auction with the hidden price as a reserve. Indeed, once the type reports are common knowledge, so is the highest bidder. There are then three possible scenarios:
\begin{itemize}
    \item The reserve price is above the highest bid. In this case, the seller proves this without revealing the hidden reserve, as in the single-buyer setting, and no trade takes place.
    \item The reserve price is weakly below the highest bid, and strictly above the second-highest bid. In this case, the seller reveals the reserve price as in the single-buyer setting, and trade takes place at this price.
    \item The reserve price is weakly below the second-highest bid. In this case, the seller proves this without revealing the hidden reserve, and trade takes place with the second-highest bid as price. While the inequality to be proved here is in the reverse direction to that proved in the single-buyer setting, an analogous technique still works: if the (hidden) price is $s=s_1\ldots s_{\log_2H}$ and the second-highest bid is $v^{(2)}=v^{(2)}_1\ldots v^{(2)}_{\log_2H}$, then the seller proves that for each $i$ such that $v^{(2)}_i=0$, either $s_i=0$ or else $s_j=0$ for some $j<i$ with $v^{(2)}_j=1$.
\end{itemize}
Since it is commonly known that only the highest bidder might participate in trade, the seller can reveal/prove the above either to all bidders or only to the highest bidder.

\subsection{Illustrative Example 2: Never Fully Revealing All Details}

A simplifying feature of Illustrative Example 1 is that if the item is sold, then the mechanism is fully revealed. Our next example is arguably one of the simplest within an auction setting where the mechanism is never fully revealed. Here we have a seller offering not one but two items to a unit-demand buyer, each at some price. The buyer has a private value $v_i$ for each item $i\in\{1,2\}$, and the seller wishes to commit to two prices $s^1,s^2\in\{0,\ldots,H\!-\!1\}$ for a commonly known $H$. Sale takes place if $v_i\ge s^i$ for some $i$, and the item sold is an item whose sale maximizes the buyer's utility.

The requirement here is that the price of any unsold item should remain hidden. (Since the buyer has unit demand, there is always at least one unsold item.) More specifically, the buyer should learn nothing about any price before disclosing $v$. If no trade takes place, then the buyer should learn that $s^1>v_1$ and $s^2>v_2$ but learn no more than that. If some trade does take place, then the buyer should only learn about the price of the unsold item that it is high enough so that if the buyer were to purchase that item at that price, then her resulting utility would be no greater than that from buying the item actually sold at its price, that is, $s^i\ge s^j-v_j+v_i$, where $i$ is the unsold item and $j$ is the sold item.

The construction in this case is very similar to that of Illustrative Example~1: The seller commits to $2\log_2H$ bits (each bit of each of the two prices) using the same scheme. After the buyer discloses her values $v_1,v_2$, if any item is sold then the seller reveals its price (along with the random values used to create the commitments to it), and for any item $i$ not sold, the seller proves that $s^i$ is weakly greater than some commonly known value: either that $s^i\ge v_i+1$ if the other item is not sold, or that $s^i\ge s^j-v_j+v_i$ if the other item $j$ is sold. This is done using the same method as in Illustrative Example 1.

\subsection[Illustrative Example 3: Proof of Incentive Compatibility]{Illustrative Example 3: Proof of Incentive Compatibility}\label{example-ic}

A simplifying feature of Illustrative Examples 1 and 2 is that every mechanism that could be committed to using the schemes in these examples is inherently IC (and IR). Our next example is arguably one of the simplest within an auction setting that do not exhibit this feature. Here we again have a seller offering a single item to a buyer, but not at a single price. Rather, there is a price $s^1$ for the initial $0.5$ probability of winning the item, and a price $s^2$ for the remaining $0.5$ probability of winning. More specifically, the mechanism is defined as follows, where $v$ is the value reported by the buyer:
\begin{itemize}
    \item If $\nicefrac{v}{2} < s^1$, then the buyer gets nothing and pays nothing.
    \item Otherwise, if $\nicefrac{v}{2} < s^2$ then the buyer gets the item with probability $0.5$ and pays $s^1$. (In this example, we assume the lottery that awards the item with probability $0.5$ is held publicly. This need not be the case---see the next example.)
    \item Otherwise, the buyer gets the item with probability $1$ and pays $s^1+s^2$.
\end{itemize}

\begin{lemma}
This mechanism is IC if and only if $s^1\le s^2$.
\end{lemma}

\begin{proof}
If $s^1\le s^2$, then $\nicefrac{v}{2}<s^1$ implies $v<s^1+s^2$, so case analysis shows that the buyer always gets the utility-maximizing outcome by reporting truthfully. Otherwise, i.e., if $s^1>s^2$, then $s^1+s^2<2s^1$, so for every $v\in(s^1+s^2,2s^1)$, the buyer has an incentive to misreport by claiming her value is $2s^1$, since truthful reporting gets her nothing (utility $0$) but reporting $2s^1$ gets her the item for sure (utility $v-(s^1+s^2)>0$).
\end{proof}

Given this lemma, after the seller commits to $s^1$ and~$s^2$, to learn that the mechanism is IC, we want the buyer to learn that $s^1\le s^2$ but learn no more than that.\footnote{This would teach the buyer that the mechanism is IC for all buyer values. The IC condition only for values in $\{0,\ldots,H\!-\!1\}$ is slightly more nuanced and not used for simplicity, but can be handled similarly.} Analogously to the previous two examples, if the buyer pays nothing and gets nothing, then the buyer should learn that $\nicefrac{v}{2} < s^1$ but learn no more than that; if the buyer gets the item with probability $0.5$ for a price of $s^1$, then the buyer should learn the value of $s^1$ and that $\nicefrac{v}{2} < s^2$ but learn no more than that; and if the buyer gets the item with probability~$1$ for a price of $s^1+s^2$, then the buyer should learn the value of $s^1+s^2$, but learn no more than that.

Most of the construction in this case is again similar to that of the previous examples. We use the same commitment scheme to commit to every bit of $s^1$ and $s^2$. Each of the inequalities that might have to be proven after the outcome is chosen (i.e., the inequalities showing that one of the prices is greater than the then-commonly-known $\nicefrac{v}{2}$) can be proven by the seller as in the previous examples, and any price can be revealed as in the previous examples. 
There are therefore two added challenges: proving the inequality $s^1\le s^2$ between the two hidden prices (which proves IC), and proving what $s^1+s^2$ is without revealing any additional information about $s^1$ or $s^2$. We will now explain how to do each of these. Let $s^1=s^1_1\ldots s^1_{\log_2H}$ and $s^2=s^2_1\ldots s^2_{\log_2H}$ be the binary representations of the two prices and let $C^1_1,\ldots,C^1_{\log_2H},C^2_1,\ldots,C^2_{\log_2H}$ be the respective commitments to the bits.

We start with proving that $s^1\le s^2$. Similarly to the proofs above, it suffices to show that for every $i\in\{1,\ldots,\log_2H\}$, either $s^1_i\le s^2_i$ or there exists $j<i$ such that $s^1_j< s^2_j$. That is, the seller should prove that for every $i$, one of the following holds: (a)~$s^1_i=0$, or (b)~$s^2_i=1$, or (c)~there exists $j<i$ such that both $s^1_j=0$ and $s^2_j=1$. That is, for every such~$i$ the seller should prove that she knows one of the following three: (a)~the discrete log base $g$ of $C^1_i$, or (b)~the discrete log base $h$ of $C^2_i$, or (c)~for some $j<i$, both the discrete log base $g$ of $C^1_j$ and the discrete log base $h$ of $C^2_j$. While this statement is of a more general form than the statements in the previous examples (which are of the form ``the seller knows the log base $h$ of one of a given set of elements of $G$''), in \cref{cds-new} in the supplementary material we describe a new extension, for this more general form of statement, of the procedure used for proving the simpler statements in the previous examples. This procedure maintains all of the desired properties. The seller proves the statement using this procedure.

The final challenge is proving what the value of $s^1+s^2$ is (in case the buyer gets the item with probability $1$). Let $s=s^1+s^2$ and write $s$ in binary representation: $s=s_0\ldots s_{\log_2H}$, i.e., $s=\sum_{i=0}^{\log_2H}s_i\cdot2^{\log_2H-i}$. Inductively define the ``carry'' bits $c_1,\ldots,c_{\log_2H}$ used in the calculation of $s^1+s^2$: for every $i\in\{1,\ldots,\log_2H\}$, starting with $i=\log_2 H$, we define $c_i=1$ if and only if at least two of the following three conditions hold: (a) $s^1_i=1$, (b) $s^2_i=1$, (c) both $i<\log_2H$ and $c_{i+1}=1$.

The seller commits to each of these $\log_2H$ carry bits using the same commitment scheme, and uses these commitments to prove that each step of the calculation of $s^1+s^2$ was properly executed without revealing $s^1$ or $s^2$.
To see how this can be done, consider for example how one might prove that $s_5$ and $c_5$ were properly calculated. Let $C^1_5,C^2_5,C_5',C_6'$ be the respective commitments to $s^1_5,s^2_5,c_5,c_6$ (assume that $\log_2 H\ge 6$ for the purposes of this example). We must show that $s_5=s^1_5+s^2_5+c_6\bmod2$ and that $c_5$ satisfies its above inductive definition. Thus, we must prove that $(s^1_5,s^2_5,c_5,c_6)$ is one of four specific valid quadruples $(b_{1,1},\ldots,b_{1,4}),\ldots,(b_{4,1},\ldots,b_{4,4})$ (which four quadruples are valid depends on the value of $s_5$, which the seller sends to the buyer). Defining $g_{i,j}=g$ if $b_{i,j}=0$ and $g_{i,j}=h$ otherwise, this is equivalent to showing that there exists $i\in\{1,\ldots,4\}$ for which the seller knows $(\rho_1,\ldots,\rho_4)$ such that $g_{i,1}^{\rho_1}=C^1_5$, $g_{i,2}^{\rho_2}=C^2_5$, $g_{i,3}^{\rho_3}=C_5'$, and $g_{i,4}^{\rho_4}=C_6'$. This form of statement (``there exists $i\in\{1,\ldots,k\}$ for which the seller knows $(\rho_1,\ldots,\rho_m)$ such that $g_{i,j}^{\rho_j}=C_{i,j}$ for every $j=1,\ldots,m$,'' where $g_{i,j}$ and $C_{i,j}$ are commonly known elements in $G$ for every $i,j$)
is a further generalization of the one used for proving that $s^1\le s^2$, and yet our new extension, from \cref{cds-new} in the supplementary material, of the procedure used for proving the simpler statements used so far can prove this form of statement as well.

While committing to carry bits might seem to be very specialized for calculating sums, we note that the method described for proving the value of $s^1+s^2$ is in fact completely general: It can be used to prove the proper computation of any function, taking any number of committed values as input, by simply committing to all of the ``intermediate bits'' of the computation of the function, and using our new generalized procedure from \cref{cds-new} to prove that each intermediate bit is properly calculated. Therefore, for any way of describing (single- or multi-buyer) mechanisms using bits, our new generalized procedure from \cref{cds-new} can be used to prove that the described mechanism is IR and IC (or satisfies any other property of interest), and after the buyer type is reported and the outcome announced, that the committed mechanism was properly run.\footnote{In multi-player settings in which the outcome naturally specifies an individual outcome for each buyer (e.g., what she gets and how much she pays, or what her match is), instead of announcing the entire outcome to all buyers and proving that it is indeed the result of the committed mechanism, the seller could announce to each buyer only her individual outcome, and prove to her only that it is indeed her individual part of the outcome resulting from the committed mechanism, revealing no additional information about the other parts of the outcome but maintaining credibility (cf.\ \citealp{AkbarpourL2020}) through our machinery.}

\subsection[Illustrative Example 4: Proof of Properly Running a Randomized Mechanism]{Illustrative Example 4: \texorpdfstring{\\}{ }Proof of Properly Running a Randomized Mechanism}

A simplifying feature of Illustrative Examples 1, 2, and 3 is that only the deterministic part of the mechanism is hidden, i.e., either the mechanism is deterministic and only the outcome is revealed, or it is randomized and the full outcome distribution is revealed (followed by public coin flipping to draw the realized outcome). Our next example is arguably one of the simplest within an auction setting where only the realized outcome is revealed.\footnote{\cite{EilatEM2023} identify the realized outcome of an auction with similar randomization to that of our example as more private (under a metric that they study) than any deterministic auction with the same or greater revenue. The use of randomness within mechanisms might therefore be motivated for secrecy even when a deterministic mechanism with the same revenue exists.}

Here, as in Illustrative Example 1, we again have a seller offering an item to a buyer for some hidden price $s$; however, this price is only an expected price: if trade takes place then instead of the buyer paying $s$ with certainty, the buyer pays~$H$ (recall that $H$ is an upper bound on the buyer's valuation) with probability $\nicefrac{s}{H}$ and pays $0$ otherwise.
The challenge is that $s$ should remain hidden even if trade takes place, revealing only the realized payment.\footnote{Replacing a payment with an expectation-equivalent distribution supported on one low value and one high value was utilized by \cite{EilatEM2023} for privacy and by \cite{RubinsteinZ2021} to save on communication (rather than for privacy).} The buyer should be convinced not only that there really is a commitment to the price, but also that the payment is the result of a faithful draw.

Our commitment scheme for this example is again the same, and like in Illustrative Examples 1 and 2, no proof that the mechanism is IC (or IR) is needed. Let $v$ be the buyer's reported value. If $v<s$, then the seller proves so (i.e., proves that $v+1\le s$) as in Illustrative Example 1. If, on the other hand, $v\ge s$, then the seller first proves that this holds (as in the multiple-buyer case of Illustrative Example~1). It remains to show how, given commitments to the bits of $s$, a coin with heads probability $\nicefrac{s}{H}$ can be verifiably flipped without the buyer learning anything about $s$ except for the realized outcome.

The seller draws a uniformly random number $x=x_1\ldots x_{\log_2H}\in \{0,\ldots,H\!-\!1\}$ and for each $i\in\{1,\ldots,\log_2H\}$, sends a commitment $R_i$ to $x_i$ and a commitment $R_i'$ to $1-x_i$, and proves to the buyer that she knows the log base $g$ of one of $R_i$ and $R_i'$ and the log base $h$ of one of $R_i$ and $R_i'$, without revealing any additional information.\footnote{Since it is assumed that the seller cannot know the log of any commitment in both bases, this means that the seller knows the log base $g$ of one of the commitments and the log base $h$ of the other, i.e., they are indeed commitments to some $x_i$ and its complement $1-x_i$.} The buyer then draws a uniformly random number $y=y_1\ldots y_{\log_2H}\in\{0,\ldots,H\!-\!1\}$ and sends it to the seller. For each $i\in\{1,\ldots,\log_2H\}$, the seller calculates $z_i=x_i+y_i\bmod2$. Note that since $H$ is a power of $2$, as long as either the seller or the buyer (or both) indeed draw their number uniformly at random, the number $z=z_1\ldots z_{\log_2H}$ is distributed uniformly in $\{0,\ldots,H\!-\!1\}$. Therefore, having the buyer pay $H$ if $z<s$, and $0$ otherwise, results in the required distribution. It remains for the seller to prove to the buyer whether or not $z<s$. Note that a commitment to the bits of $z$ is in fact commonly known: for each $i$, the commitment to $z_i$ is $R_i$ if $y_i=0$ and $R_i'$ if $y_i=1$. Since commitments to both $z$ and $s$ are known, the seller can use the technique from Illustrative Example 3 to prove to the buyer that $z<s$ (or alternatively, that $z\ge s$) without revealing any additional information.

We note that, similarly to before, this method is completely general for any way of describing (single- or multi-buyer) randomized mechanisms using bits: it not only allows the calculation of any function of both public bits and committed bits as in Illustrative Example 3, but furthermore allows functions that also depend on random bits in arbitrary ways.\footnote{With multiple buyers, each buyer $j$ draws a uniformly random number $y^j=y^j_1\ldots y^j_{\log_2H}\in\{0,\ldots,H\!-\!1\}$. These numbers are then bitwise-summed modulo 2 to obtain a number $y=y_1\ldots y_{\log_2H}\in\{0,\ldots,H\!-\!1\}$ that is used to shuffle the seller's commitment to the bits of the random number $x$ as in the example above. As long as at least one agent, buyer or seller, draws her number uniformly at random, then the resulting shuffled commitment is to the bits of a uniformly random number.} For example, no additional techniques are needed to implement the mechanism described in Illustrative Example 3 so that the buyer only sees the realized outcome, e.g., if the buyer gets the item, then she sees her realized payment, but never knows whether it equals $s^1$ (and she got the item due to a lucky draw) or it equals $s^1+s^2$ (and she had probability $1$ of getting the item). Having intuitively seen the properties that we guarantee, we now turn to develop the formal apparatus needed to state them precisely.

\section{Preliminaries}\label{prelims}

\subsection{Standard Definitions}\label{prelims-standard}

\paragraph{Players, types, outcomes, and utilities} There are finitely many \emph{players} $i=1,\ldots,n$. Each player $i$ has a private \emph{type} $t_i$ from a finite set $T_i$ of possible types.
There is a finite set $X$ of possible (pure) \emph{outcomes} (these include payments, if applicable). Each player $i$ has a \emph{utility function} $u_i: T_i\times X\rightarrow\mathbb{R}$, where $u_i(t_i,x)$ is the utility of player $i$ with (private) type $t_i$ from outcome $x$. We denote $T=T_1\times\cdots\times T_n$.

\begin{example}\label{nm-setting}
In an $n$-bidder $m$-item auction setting with additive bidder valuations and a commonly known upper bound $H$ on any bidder's value for any single item, we can take the outcome set to be $X = \bigl\{y\in\{0,1\}^{n\times m}~\big|~\forall j=1,\ldots,m:\sum_{i=1}^n y_{ij}\le 1\bigr\}\times \{0,\ldots,mH\}^n$.
Here, an outcome $x\in X$ is interpreted as a binary allocation matrix $y_{ij}$ (of dimensions $n\times m$) and a payment vector $p_1, \ldots, p_n$. If $y_{ij}=1$ for some $i=1,\ldots,n$ and $j=1,\ldots, m$, then bidder $i$ gets item $j$. Each bidder $i$ pays $p_i$. By construction, no item may be allocated more than once, but some items may remain unallocated. Taking each type set to be $T_i=\{0,\ldots,H\}^m$, indicating a value for each item, the utility function of each player $i$ can be taken as $i$'s value of items allocated to $i$ less $i$'s payment.
\end{example}

\paragraph{Mechanisms and incentives} A \emph{(direct-revelation) mechanism} $M$ is a function $M:T\rightarrow\Delta(X)$. A mechanism $M$ is \emph{Individually Rational (IR)} if $\Exp[u_i(t_i, M(t))]\ge0$ for every $i=1,\ldots,n$ and $t\in T$ (the expectation is over the randomness of the mechanism). A mechanism $M$ is \emph{Dominant Strategy Incentive Compatible (DSIC)} if $\Exp[u_i(t_i, M(t))]\ge\Exp[u_i(t_i, M(t_i', t_{-i}))]$ for every $i=1,\ldots,n$, $t\in T$, and $t_i'\in T_i$, where $(t_i', t_{-i})$ denotes the type profile derived from $t$ by replacing player $i$'s type to be $t_i'$ instead of $t_i$.

It will be convenient to view a randomized mechanism $M:T\rightarrow\Delta(X)$ as
a \emph{deterministic} function $M:T \times \{0,1\}^{\runrandlen} \rightarrow X$ where the second input to $M$ is a random sequence of $\runrandlen$ \emph{run bits} drawn uniformly from $\{0,1\}^{\runrandlen}$ for a suitable $\runrandlen$.\footnote{In practice, any algorithmic implementation of a randomized algorithm always uses only finitely many bits of randomness.}\textsuperscript{,}\footnote{\label{re-permutation}There are many different ways to map sequences of $\runrandlen$ random bits to outcomes in the support of the outcome distribution of the mechanism. Indeed, taking a function $M:T\times\{0,1\}^{\runrandlen}$ and applying to its second input any bijection of $\{0,1\}^{\runrandlen}$ onto itself results in a different function $M':T\times\{0,1\}^{\runrandlen}$ that nonetheless represents the same randomized mechanism (when viewed as a function $M:T\rightarrow\Delta(X)$). This fact is important since we will want knowledge of $t\in T$, of $\runrand\in\{0,1\}^{\runrandlen}$, and of the outcome $x=M(t,\runrand)$ to not leak any information about the \emph{distribution} $F=M(t)\in\Delta(X)$ beyond the fact that $x$ is in the support of $F$. We return to this point later.} Hereinafter we always describe a randomized mechanism in this way.

\begin{example}
In a setting with two bidders and one item, with $X=\{0,1\}^2\times\{0,\ldots,H\}^2$ and $T_1=T_2=\{0,\ldots,H\}$ as in \cref{nm-setting}, a second-price auction with no reserve price and random tie-breaking is the following randomized mechanism:
\[
M(t_1,t_2) =
\begin{cases}
(1,0,t_2,0) & t_1>t_2, \\
(0,1,0,t_1) & t_2>t_1, \\
\text{each of the above with probability $0.5$} & t_1=t_2.
\end{cases}
\]
This mechanism, when viewed as a deterministic function taking a random sequence of run bits as an additional input, can be written as follows (it suffices to take $\runrandlen=1$):\pagebreak[0]
\[
M(t_1,t_2,\runrand) =
\begin{cases}
(1,0,t_2,0) & t_1>t_2~\text{or both $t_1=t_2$ and (the first bit of) $\runrand$ is $0$},\\
(0,1,0,t_1) & \text{otherwise}.\\
\end{cases}
\]
\end{example}

\subsection{Languages for Mechanisms and Proofs}

\paragraph{Mechanism-description languages and mechanism interpreters} To formally reason about how mechanism designers describe mechanisms to players, one needs to choose a \emph{mechanism-description language}, i.e., a way to unambiguously represent mechanisms. Depending on the richness of the family of mechanisms the designer might wish to be able to describe, the mechanism-description language might be more or less complex. It is convenient to focus on a mechanism-description language that represents every mechanism of interest as a finite sequence of bits.

\begin{example}\label{languages}
There is a wide range of possible mechanism-description languages, e.g.:
\begin{itemize}
\item 
The Python programming language is a very general mechanism-description language. In this mechanism-description language, a specific mechanism is represented by the code of a Python function that takes the player types and run bits as inputs, and outputs the outcome.\footnote{The same mechanism can be described in more than one way in this language. Our analysis does not in any way require that each mechanism have a unique description.}
\item
If one is interested in describing only mechanisms of a specific functional form, then the mechanism-description language can be more specific. For example, a language used to describe second-price auctions with a reserve price can represent a specific such auction using the reserve price (written out in base $2$).
\item More generally, a mechanism-description language for a parametrized family of mechanisms can describe each mechanism using its parametrization.
E.g., we might describe any mechanism of the form from Illustrative Example~3 in~\cref{example-ic} using the two prices $s^1$ and $s^2$ (each written out in base $2$ using the same number of bits).
\end{itemize}
\end{example}

We formalize the notion of a mechanism-description language by specifying how the language is interpreted. A \emph{(mechanism) interpreter} (for a specific mechanism-description language) is a function $I:\{0,1\}^*\times T\times \{0,1\}^*\times\mathbb{N}\rightarrow X\cup\{\text{``error''}\}$, where we adopt the computer-science notation $\{0,1\}^*=\bigcup\bigl\{\{0,1\}^n~\big|~n\in\mathbb{N}\bigr\}$ denoting the set of all finite sequences of bits. The four inputs of a mechanism interpreter~$I$ are: (1) a \emph{mechanism description} $A\in \{0,1\}^*$ (in the language corresponding to the interpreter), (2) a type profile $t\in T$, (3) a sequence of run bits $\runrand\in\{0,1\}^*$ (treated as random bits as described in \cref{prelims-standard}), and (4) a maximum \emph{running time} (number of computation steps) $R\in\mathbb{N}$. We explicitly assume that given such $A$, $t$, and $\runrand$, there is an abstract, objective way to measure the number of basic computational steps (the ``running time'') required to run the mechanism according to the description $A$ over the type profile $t$ with run bits $\runrand$.\footnote{This high-level assumption suffices for our purposes and saves us from defining formal notions from the theory of programming languages, such as ``execution models.''}
The output $I(A,t,\runrand,R)$ of an interpreter is the outcome of running the mechanism according to the description $A$ over type profile $t$ with run bits $\runrand$, as long as this takes no more than $R$ running time---or the indication of an error if either $A$ is not a valid mechanism description, not enough run bits are specified, or the maximum running time is exceeded.

\begin{example}\label{interpreters}
The mechanism-description languages given in \cref{languages} can be formalized using the following interpreters:
\begin{itemize}
    \item An interpreter for the Python-based mechanism-description language from \cref{languages} indicates an error if either $A$ is not a valid Python function code, it crashes, it does not complete running within $R$ basic Python computation steps, or not enough run bits for $A$ are given in $\runrand$. Otherwise, it outputs the outcome.
    \item An interpreter for the mechanism-description language for second-price auctions with a reserve price from \cref{languages} ignores the run bits, outputs the outcome of the auction, and never indicates an error.
    \item An interpreter for the third mechanism-description language from \cref{languages} uses a single run bit for randomness, i.e., if $s^1\le \nicefrac{t}{2}<s^2$, then the outcome is $(0,s^1)$ if the run bit is $0$, and $(1,s^1)$ if the run bit is $1$. This interpreter indicates an error if $A$ does not consist of an even number of bits or the sequence of run bits $\runrand$ has length $0$.
\end{itemize}
\end{example}
For the remainder of this paper, we fix some specific mechanism interpreter~$I$ (and associated mechanism-description language). For concreteness, one may consider any interpreter from \cref{interpreters}. Given a mechanism description $A$ and $\runrandlen,R\in\mathbb{N}$, if for all type profiles $t\in T$ and run bits $\runrand\in\{0,1\}^{\runrandlen}$ it holds that $I(A,t,\runrand,R)\ne\text{``error''}$ (i.e., if $A$ is a valid mechanism description that always uses at most $\runrandlen$ run bits and runs in time at most $R$), then we say that $A$ is \emph{$(\runrandlen,R)$-runnable}, and denote by $M_A$ the mechanism described by $A$ (as interpreted by~$I$), i.e., $M_A(t,\runrand)=I(A,t,\runrand,R)$ (where $\runrand\in\{0,1\}^{\runrandlen}$).

\paragraph{Formal proof language} We assume a common language in which we can unambiguously state and prove mathematical claims. In particular, we assume the ability to formally state and prove claims of the forms ``$A$ is an $(\runrandlen,R)$-runnable description of an IR and DSIC mechanism'' and ``$A$, when run on type profile $t$ using run bits $\runrand$, has outcome $x$.'' We assume that the length of a proof for the latter statement is proportional to the running time of $A$, and that verifying a proof can be done in running time proportional to the sum of the lengths of the statement and the proof.

For concreteness, one could use predicate logic \citep[see, e.g.,][]{GonczarowskiN2022} or the Lean programming language (\url{https://lean-lang.org/}) as the formal language for stating and proving claims (both of these satisfy the above assumptions with respect to any real-life running times). For ease of exposition, we will not explicitly construct statements and proofs in either of these languages; however, all claims and proofs we require in this paper can be written in them. Moreover, any claim or proof from any textbook can be written in these languages, and its length will be proportional to its length in the textbook.
We will not need any further details for the purposes of this paper.
For an example of a formal proof of the incentive properties of a mechanism, see, e.g., \cite{Barthe16}.

\section{A Framework for Mechanism Hiding}\label{framework}

In this section, we formally develop a framework for mechanism hiding without a mediator. In \cref{protocols}, we define a general parametrized class of structured, sequential messaging environments between a mechanism designer and players, which we call \emph{direct-revelation commit-and-run protocols}. In \cref{desiderata}, we formulate our desiderata for direct-revelation commit-and-run protocols, which capture them having the same desirable trust and secrecy properties as running a mechanism through a trustworthy mediator, despite the absence of any mediator (trusted or otherwise).

\subsection{Commit-and-Run Protocols}\label{protocols}

In this section, we define commit-and-run protocols with a focus on the mechanics of running them; we postpone the discussion of the semantic meaning of the various parts of these protocols to \cref{desiderata}, in which we formalize the desiderata for our protocols.

\paragraph{Specifications}
Each of our protocols is parametrized by several bounds on the mechanism descriptions that it is able to run (where every combination of bounds is supported by one of our protocols): (1)~an upper bound $L_a\in\mathbb{N}$ on description lengths, (2)~an upper bound $\runrandlen\in\mathbb{N}$ on quantities of (uniformly drawn) run bits used, (3)~an upper bound $R\in\mathbb{N}$ on running times, and (4)~an upper bound $L_p\in\mathbb{N}$ on lengths of proofs that a description is an $(\runrandlen,R)$-runnable description of an IR and DSIC mechanism. As discussed in the introduction, each of our protocols is also parametrized by an upper bound $B\in\mathbb{N}$ on the computational power of an \emph{adversary} (mechanism designer who attempts to break her commitment, or players attempting to discover a hidden mechanism) and by an upper bound $0<\varepsilon<1$ on the attack-success probability of any adversary. As for the units of~$B$, we formally measure computational power in terms of available running time: The higher an agent's computational power, the greater the number of computation steps she can run, i.e., the greater the running time at her disposal. The precise meaning of these bounds is defined in \cref{desiderata}.

Overall, we call a tuple $\sigma=(L_a,\runrandlen,R,L_p,B,\varepsilon)$ of such parameters, where $R\ge L_a$, a \emph{specification}. Let $\Sigma$ be the set of all specifications such that at least one $(\runrandlen,R)$-runnable description, of length at most $L_a$, of an IR and DSIC mechanism exists.
We prove in later sections that a protocol satisfying our desiderata exists for every specification $\sigma\in\Sigma$.

\paragraph{Protocols}
A (direct-revelation commit-and-run) protocol consists of two parts: the grammar rules and the playbook. The grammar rules define the extensive-form interaction, including the order in which messages are sent, the length of each message, and the publicly verifiable required relation of messages to each other. The playbook defines a ``recommended behavior'' for the mechanism designer.\cryptofootnote{In presenting a protocol using these two components, we deviate from the standard computer-science view of protocols in favor of greater usability for economists. Indeed, computer science usually views a protocol as a set of instructions for each agent, where the instructions contain both what the agent is expected to send at any given moment (what we call the playbook) and which test to run on anything received (which implicitly defines what we call the grammar rules).} We now define each of these.

\begin{table}[p]
\begin{center}
{\footnotesize
\begin{tabular}{lll}
\textbf{Notation} & \textbf{Meaning} \\
\hline
\underline{Setting} \\
$n$ & number of players \\
$X$ & possible outcomes \\
$T$ & possible type profiles \\
$M:T\rightarrow\Delta(X)$ \quad / \quad $M:T\times\smash{\underbrace{{\{0,1\}^{\runrandlen}}}_{\mathrm{uniform}}}\rightarrow X$ & direct-revelation mechanism\\ $\vphantom{\underbrace{a}{a}}$ & (intrinsic/externally given randomness)\\
\hline
\underline{Mechanisms and Proofs} \\
$A$ & mechanism description \\
$\mathcal{A}_{L_a,\runrandlen,R}$ & $(\runrandlen,R)$-runnable mechanism descriptions of length $\le L_a$ \\
$M_A$ & mechanism defined by description $A$ \\
$P$ & proof \\
$\mathcal{P}_{L_p}$ & proofs of length at most $L_p$ \\
\hline
\underline{Specifications} \\
$L_a$ & upper bound on mechanism description length \\
$\runrandlen$ & upper bound on number of run bits used \\
$R$ & upper bound on mechanism running time \\
$L_p$ & upper bound on length of proof that a description is \\
& $(\runrandlen,R)$-runnable and of an IR and DSIC mechanism \\
$B$ & upper bound on adversary running time \\
$\varepsilon$ & upper bound on adversary success probability \\
$\sigma=(L_a,\runrandlen,R,L_p,B,\varepsilon)\in\Sigma$ & specification \\
\hline
\underline{Protocol Grammar Rules} \\
$\commitrandlen$ & number of (common) commitment random bits \\
$\commitmsglen$ & length in bits of commitment message \\
$\commitverify:\{0,1\}^{\commitrandlen}\times\{0,1\}^{\commitmsglen}\rightarrow\{\True,\False\}$ & commitment verifier \\
$\runrandlen$ & number of (common) drawn run bits \\
$\runmsglen$ & length in bits of run message \\
$\runverify:\{0,1\}^{\commitrandlen}\times\{0,1\}^{\commitmsglen}\times T\times \{0,1\}^{\runrandlen} $ & run verifier \\
\qquad $\phantom{a}\times \bigl(X\times \{0,1\}^{\runmsglen}\bigr)\rightarrow\{\True,\False\}$ \\
\hline
\underline{Running a Protocol} \\
$\commitrand\in\{0,1\}^{\commitrandlen}$ & (common) commitment random bits \\
$\commitmsg\in\{0,1\}^{\commitmsglen}$ & commitment message (mnemonic for $D$eclaration) \\
$\commitstrat:\{0,1\}^{\commitrandlen}\rightarrow\Delta\bigl(\{0,1\}^{\commitmsglen}\times\{0,1\}^*\bigr)$ & mechanism designer commitment strategy \\
$s\in\{0,1\}^*$ & mechanism-designer private state \\
$\runrand\in\{0,1\}^{\runrandlen}$ & (common) run bits (mechanism need not use all) \\
$\runmsg\in\{0,1\}^{\runmsglen}$ & run message \\
$\runstrat:\{0,1\}^{\commitrandlen}\times\bigl(\{0,1\}^{\commitmsglen}\times\{0,1\}^*\bigr)$ & mechanism-designer run strategy\\
\qquad $\phantom{a}\times T\times \{0,1\}^{\runrandlen}
\rightarrow\Delta\bigl(X\times\{0,1\}^{\runmsglen}\bigr)$ \\
$(\commitstrat,\runstrat)$ & mechanism-designer strategy \\
\hline
\underline{Playbooks} \\
$\protcommitstrat:\mathcal{A}_{L_a,\runrandlen,R}\times \mathcal{P}_{L_p}\times\{0,1\}^{\commitrandlen}$ & commitment playbook (recommended commitment \\
\qquad $\rightarrow\Delta\bigl(\{0,1\}^{\commitmsglen}\times\{0,1\}^*\bigr)$
& strategy for well-behaved mechanism designer) \\
$\protrunstrat:\mathcal{A}_{L_a,\runrandlen,R}\times\{0,1\}^{\commitrandlen}\times\bigl(\{0,1\}^{\commitmsglen}\times\{0,1\}^*\bigr)$& run playbook (recommended run strategy \\
\qquad $\phantom{a}\times T\times \{0,1\}^{\runrandlen}
\rightarrow\Delta\bigl(X\times\{0,1\}^{\runmsglen}\bigr)$ & for well-behaved mechanism designer)\\
\hline
\underline{Protocols} \\
$\bigl((\commitrandlen,\commitmsglen,\commitverify,\runrandlen,\runmsglen,\runverify), (\protcommitstrat,\protrunstrat)\bigr)$ & protocol (grammar rules and playbook) \\
$\bigl((\commitrandlen^\sigma,\commitmsglen^\sigma,\commitverify^\sigma,\runrandlen^\sigma,\runmsglen^\sigma,\runverify^\sigma), (\protcommitstrat^\sigma,\protrunstrat^\sigma)\bigr)_{\sigma\in\Sigma}$ & protocol catalog (protocol for each specification)
\end{tabular}}
\end{center}
\caption{\label{notation}Notation for direct-revelation commit-and-run protocols.}
\end{table}

\paragraph{Grammar rules} A protocol with \emph{grammar rules} $(\commitrandlen,\commitmsglen,\commitverify,\runrandlen,\runmsglen,\runverify)$ (defined below; see also \cref{notation}) proceeds as follows:
\begin{enumerate}
    \item[0.] Commitment randomness generation step: nature draws a uniformly random sequence $\commitrand\in\{0,1\}^{\commitrandlen}$ (the number of bits $\commitrandlen$ is part of the grammar rules) and publicly announces it.\footnote{There are various well-documented methods to obtain random (or ``sufficiently random'') bits even if the different agents are far from one another, from basing these bits on publicly available data that are not manipulable by any of the agents (as was done in practice in the ``numbers racket'' operated by organized crime in New York City in the 1960s, \citealp[p.~68]{Cook14}, and in a more sophisticated way to create the randomness for the first Bitcoin block, \citealp{Davis2011}), through observing rapidly changing physical phenomena (e.g., the NIST Randomness Beacon, \url{https://beacon.nist.gov/home}), to cryptographic coin flipping \citep{Blum81,GMW87}. In this paper, we abstract away from the precise method.}
    \item[1a.] Commitment step: the mechanism designer, given $\commitrand$, chooses a commitment message $\commitmsg\in\{0,1\}^{\commitmsglen}$ (the number of bits $\commitmsglen$ is part of the grammar rules) and publicly announces it.
    \item[1b.] Commitment verification step: each player can evaluate $\commitverify(\commitrand,\commitmsg)\in\{\True,\False\}$ (the predicate $\commitverify$ is part of the grammar rules) to verify that it holds (i.e., returns $\True$).\footnote{As we will see, our desiderata in \cref{desiderata} mandate that if $\commitverify(\commitrand,\commitmsg)$ holds then the players are convinced that the mechanism designer has committed to a description $A$ of an IR and DSIC mechanism, despite knowing nothing more about the mechanism than if it had been run via a trustworthy mediator.}
    \item[2.] Direct revelation step: each player publicly reveals her type. Let $t\in T$ be the profile of revealed types.
    \item[3.] Run randomness generation step: nature draws a uniformly random sequence $\runrand\in\{0,1\}^{\runrandlen}$ (the possibly-zero number $\runrandlen$ of bits is part of the grammar rules) and publicly announces it.
    \item[4a.] Run step: the mechanism designer, given $t$ and $\runrand$, publicly announces an outcome $x\in X$ as the outcome of the mechanism, and also chooses a run message $\runmsg\in\{0,1\}^{\runmsglen}$ (the number of bits $\runmsglen$ is part of the grammar rules) and publicly announces it.
    \item[4b.]
    \begin{sloppypar}
    Run verification step: each player can evaluate $\runverify\bigl(\commitrand,\commitmsg,t,\runrand,(x,\runmsg)\bigr)\in\{\True,\False\}$ (the predicate $\runverify$ is part of the grammar rules) to verify that it holds.\footnote{As we will see, our desiderata in \cref{desiderata} mandate that if $\runverify\bigl(\commitrand,\commitmsg,t,\runrand,(x,\runmsg)\bigr)$ holds then the players are convinced that $x=M_A(t,\runrand)$ where $A$ is the mechanism fixed in step 1a (as verified in step 1b), despite knowing nothing more about the mechanism than if it had been run via a trustworthy mediator.}
    \end{sloppypar}
\end{enumerate}

It will be convenient to think of a mechanism designer's strategy as a pair of strategies: one for the commitment step and one for the run step. A (mixed) \emph{commitment strategy} is a randomized function $\commitstrat:\{0,1\}^{\commitrandlen}\rightarrow\Delta\bigl(\{0,1\}^{\commitmsglen}\times\{0,1\}^*\bigr)$ from the public bits $\commitrand$ from step 0 to (a joint distribution over) a commitment message $\commitmsg$ and a \emph{mechanism-designer private state}~$s$. A (mixed) \emph{run strategy} is a randomized function $\runstrat:\{0,1\}^{\commitrandlen}\times\bigl(\{0,1\}^{\commitmsglen}\times\{0,1\}^*\bigr)\times T\times\{0,1\}^{\runrandlen}\rightarrow\Delta\bigl(X\times \{0,1\}^{\runmsglen}\bigr)$ from the public bits $\commitrand$ from step 0, the commitment message $\commitmsg$ from step 1a and the mechanism-designer private state $s$ returned by the commitment strategy, the type profile $t$ from step 2, and the run bits $\runrand$ from step 3, to (a joint distribution over) an outcome $x$ and a run message $\runmsg$. A \emph{mechanism-designer strategy} is a pair $(\commitstrat,\runstrat)$ of commitment and run strategies. The mechanism-designer private state allows the random choices of the commitment strategy and run strategy to be correlated.

\paragraph{Playbooks}
Given a mechanism description, a playbook specifies recommended instructions (the \emph{commitment playbook}) for constructing the commitment message and recommended instructions (the \emph{run playbook}) for constructing the run message. We now define these formally. Let $(L_a,\runrandlen,R,L_p,B,\varepsilon)\in\Sigma$ be a specification. We write
$\mathcal{A}_{L_a,\runrandlen,R}$ for the set of all $(\runrandlen,R)$-runnable mechanism descriptions of length at most $L_a$. We write $\mathcal{P}_{L_p}$ for the set of all proofs of length at most~$L_p$.

A \emph{commitment playbook} is a function $\protcommitstrat:\mathcal{A}_{L_a,\runrandlen,R}\times \mathcal{P}_{L_p}\times\{0,1\}^{\commitrandlen}\rightarrow\Delta\bigl(\{0,1\}^{\commitmsglen}\times\{0,1\}^*\bigr)$. Given a mechanism description $A$ and a (valid) proof $P$ that $A$ is an $(\runrandlen,R)$-runnable description of an IR and DSIC mechanism, we think of such a playbook as recommending that a mechanism designer who wishes to run~$A$ use the commitment strategy $\protcommitstrat(A,P,\cdot)$ obtained from $\protcommitstrat$ by fixing the first two parameters of the playbook to be $A$ and $P$. In other words, given $A$, $P$, and the public commitment random bits $\commitrand$, the playbook recommends that the mechanism designer send a commitment message and store a private state jointly distributed according to $\protcommitstrat(A,P,\commitrand)$.\footnote{As we will see, one of our desiderata in \cref{desiderata} is that if a mechanism designer does not effectively follow this recommendation for some $A$ and appropriate $P$, then the predicate $\commitverify$ will evaluate to $\False$, and so this will become known to the players.}

A \emph{run playbook} is a function $\protrunstrat:\mathcal{A}_{L_a,\runrandlen,R}\times\{0,1\}^{\commitrandlen}\times\bigl(\{0,1\}^{\commitmsglen}\times\{0,1\}^*\bigr)\times T\times \{0,1\}^{\runrandlen}\rightarrow\Delta\bigl(X\times \{0,1\}^{\runmsglen}\bigr)$. Given a mechanism description $A$, we think of such a playbook as recommending that a mechanism designer who wishes to run~$A$ use the run strategy $\protrunstrat(A,\cdot,\cdot,\cdot,\cdot)$ obtained from $\protrunstrat$ by fixing its first parameter to be $A$. In other words, given $A$, the mechanism-designer private state $s$,
and all information $\commitrand,\commitmsg,t,\runrand$ publicly observed up to the run step, the playbook recommends that the mechanism designer send an outcome and a run message jointly distributed according to $\protrunstrat\bigl(A,\commitrand,(\commitmsg,s),t,\runrand\bigr)$.\footnote{Similarly to the commitment playbook, as we will see, one of our desiderata below is that if a mechanism designer does not effectively follow this recommendation for the same $A$ that she used in the commitment step, then the predicate $\runverify$ will evaluate to $\False$, and so this will become known to the players. In our run playbook, the outcome to be sent will of course be $M_A(t,\runrand)$; we define the run playbook as we do because it is convenient for it to specify an entire run strategy rather than only the message part of it.}

Finally, a \emph{playbook} is a pair $(\protcommitstrat,\protrunstrat)$ of commitment and run playbooks.\cryptofootnote{An object that recommends cryptographic messages is sometimes called a \emph{compiler} in the cryptographic literature.} We call a pair $\bigl((\commitrandlen,\commitmsglen,\commitverify,\runrandlen,\runmsglen,\runverify),(\protcommitstrat,\protrunstrat)\bigr)$ of grammar rules and corresponding playbook a \emph{(direct-revelation commit-and-run) protocol}.\footnote{As we will see, one of our desiderata in \cref{desiderata} is that if the mechanism designer follows the playbook with mechanism description $A$ and appropriate $P$, then $A$ is run and $\commitverify$ and $\runverify$ evaluate to $\True$.}

\subsection{Desiderata}\label{desiderata}

In this section, we define our desiderata for a (direct-revelation commit-and-run) protocol to provide the same desirable trust and secrecy guarantees as running a mechanism via a trustworthy mediator. We have four such desiderata: we require that the protocol be \emph{implementing}, \emph{committing}, \emph{hiding}, and \emph{computationally unobtrusive}.

\paragraph{Implementing} We start by defining what it means for a protocol to be ``implementing.'' In a nutshell, if the mechanism designer follows the playbook, then the outcome of the protocol should be the outcome of running the mechanism designer's mechanism description with the players' type reports, and all verifications should pass.

\begin{definition}[Implementing Protocol\cryptofootnote{Unlike the standard cryptographic notion of ``completeness,'' our ``implementing'' desideratum is not probabilistic.}]\label{implementing}
A protocol $\bigl((\commitrandlen,\commitmsglen,\commitverify,\runrandlen,\runmsglen,\runverify),(\protcommitstrat,\protrunstrat)\bigr)$
is \emph{implementing} for a specification $(L_a,\runrandlen,R,L_p,B,\varepsilon)\in\Sigma$ if for every $A\in\mathcal{A}_{L_a,\runrandlen,R}$ and proof $P\in\mathcal{P}_{L_p}$ that $A$ is an $(\runrandlen,R)$-runnable description of an IR and DSIC mechanism, and for every type profile $t\in T$ and run bits $\runrand\in\{0,1\}^{\runrandlen}$, the following holds (for the grammar rules $(\commitrandlen,\commitmsglen,\commitverify,\runrandlen,\runmsglen,\runverify)$). If the mechanism designer plays the strategy $\bigl(\protcommitstrat(A,P,\cdot),\protrunstrat(A,\cdot,\cdot,\cdot,\cdot)\bigr)$ and the players play~$t$, then both of the following hold:
    \begin{itemize}
        \item Verifications pass, i.e., $\Pr_{\commitrand\sim U(\{0,1\}^{\commitrandlen}),\protcommitstrat,\protrunstrat}\bigl[\commitverify\bigl(\commitrand,\commitmsg\bigr)\ \&\ \runverify\bigl(\commitrand,\commitmsg,t,\runrand,(x,\runmsg)\bigr)\bigr]=1$ for $(\commitmsg,s)\sim\protcommitstrat(A,P,\commitrand)$ and $(x,\runmsg)\sim \protrunstrat\bigl(A,\commitrand,(\commitmsg,s),t,\runrand\bigr)$.
        \item The mechanism that is run is $M_A$, i.e., $\Pr_{\commitrand\sim U(\{0,1\}^{\commitrandlen}),\protcommitstrat,\protrunstrat}\bigl[x=M_A(t,\runrand)\bigr]=1$ for $x$ distributed as just defined (and where if $M_A$ requires less than $\runrandlen$ random bits, then it only uses an appropriate prefix of $\runrand$).
    \end{itemize}
\end{definition}

\paragraph{Committing} Next, we define what it means for a protocol to be ``committing.'' In a nutshell, whether or not the mechanism designer follows the playbook, either her commitment message should give rise to an IR and DSIC mechanism whose output will be the outcome,\footnote{Recall from the introduction that for simplicity, we focus on IR and DSIC as the properties of interest of the mechanism that are to be verified before it is run. Our framework can be made to support arbitrary properties of mechanisms (given appropriate proof $P$ for the playbook); see \cref{discussion} for a discussion.} or at least one of the verifications should fail with probability at least $1-\varepsilon$.

Notice that the timing of the interaction in a commit-and-run protocol allows the mechanism designer to choose the mechanism based on the commitment random bits~$\commitrand$ together with additional private randomness. This is harmless as long as this randomness is realized before---and independently of---the player types and run bits $\runrand$. (Indeed, a distribution over IR and DSIC mechanisms is itself an IR and DSIC mechanism.) The following definition says that if all verifications pass, then the mechanism is effectively chosen by a randomized function $\alpha$ taking as input only the commitment bits $\commitrand$ (and hence not depending on anything else that the mechanism designer learns during the run of the protocol). The only exception might possibly be with negligibly small (at most $\varepsilon$) probability over the draw of $\commitrand$ and the randomness of $\alpha$ (e.g., when the strategy contains or successfully guesses discrete logarithms for these specific~$\commitrand$).
Note that our definition needs to capture the possibility that the mechanism chosen by the function $\alpha$ is jointly distributed with the action prescribed by the commitment strategy. Thus, formally, there should be a \emph{coupling} of (the randomness of) the commitment strategy and (the randomness of) the function $\alpha$ so that with probability $1\!-\!\varepsilon$ their outputs align.

\begin{definition}[Committing Protocol\cryptofootnote{One may contrast our ``committing'' desideratum with the standard cryptographic notion of ``efficient extraction,'' which, adapted to our setting, would have required, roughly speaking, the ability to approximately sample in a computationally feasible manner from pairs of random bits and mechanism descriptions such that verifications pass with high probability. (In particular, note that we do not require that $\alpha$, or $\Xi$, be feasibly computable.) From a cryptographic viewpoint, our proof in \cref{proofs} that the protocols that we construct satisfy our ``committing'' desideratum provides new motivation for a property of the cryptography underlying our construction that is non-standard in the cryptographic literature: having the real and simulated reference strings be drawn from the same distribution; see \cryptocref{cp-proof-sec} in the supplementary material for a discussion.}]\label{committing}
A protocol with grammar rules \linebreak $(\commitrandlen,\commitmsglen,\commitverify,\runrandlen,\runmsglen,\runverify)$ is \emph{committing} for a specification $(L_a,\runrandlen,R,L_p,B,\varepsilon)\in\Sigma$ if for every mechanism-designer strategy $(\commitstrat,\runstrat)$ computable in time at most $B$, there exist:
\begin{enumerate}
    \item a~mechanism-choice function $\alpha:\{0,1\}^{\commitrandlen}\to\Delta(\mathcal{A}_{L_a,\runrandlen,R})$, where for every $\commitrand\in\{0,1\}^{\commitrandlen}$, $\alpha(\commitrand)$ is supported solely on descriptions of IR and DSIC mechanisms, and
    \item a pointwise coupling $\Xi:\{0,1\}^{\commitrandlen}\to\Delta\bigl((\{0,1\}^{\commitmsglen}\times\{0,1\}^*)\times\mathcal{A}_{L_a,\runrandlen,R}\bigr)$ of $\commitstrat$ and $\alpha$,\footnote{This means that for every $\commitrand\in\{0,1\}^{\commitrandlen}$, the joint distribution $\Xi(\commitrand)$ is a coupling of $\commitstrat(\commitrand)$ and $\alpha(\commitrand)$; in other words, for every such $\commitrand$, the marginals of the joint distribution $\Xi(\commitrand)$ are $\commitstrat(\commitrand)$ and $\alpha(\commitrand)$.}
\end{enumerate}
such that the following holds. For every type profile $t\in T$ and run bits $\runrand\in\{0,1\}^{\runrandlen}$, with probability at least $1\!-\!\varepsilon$, if both commitment and run verifications pass then the outcome is that of $M_{A}$, where $A$ is the draw from $\alpha(\commitrand)$ under the coupling; that is, $\Pr_{\commitrand\sim U(\{0,1\}^{\commitrandlen}),\Xi,\runstrat}\bigl[\commitverify\bigl(\commitrand,\commitmsg\bigr)\ \&\ \runverify\bigl(\commitrand,\commitmsg,t,\runrand,(x,\runmsg)\bigr)\ \&\ x\ne M_A(t,\runrand)\bigr]\le \varepsilon$ for $\bigl((\commitmsg,s),A\bigr)\sim\Xi(\commitrand)$ and $(x,\runmsg)\sim \runstrat\bigl(\commitrand,(\commitmsg,s),t,\runrand\bigr)$.  (Recall that $(\commitmsg,s)$ here is distributed as if output by $\commitstrat(\commitrand)$, and $A$ here is distributed as if output by $\alpha(\commitrand)$.)
\end{definition}

\cref{committing} makes precise the contrast that we draw in the introduction: A traditional commitment to an observed mechanism via verification of publicly observed signals (the announced mechanism and outcome) is replaced with a \emph{self-policing} commitment to a \emph{never-observed mechanism}, yet still via verification of publicly observed signals (the messages).
The definition's guarantee is that whatever the mechanism-designer strategy is---even if examining its inner workings does not readily reveal any explicit mechanism---passing the verification of the public messages alone quite strikingly implies (with all but negligible probability) that some IR and DSIC mechanism is in fact implicitly chosen by the commitment strategy. This is the same guarantee as the declaration of a trustworthy mediator would provide, yet here it emerges from the messages themselves. The mechanism, whether explicit or implicit in the designer's strategy, is \emph{never observed} by anyone except the designer (not even by a trusted mediator)---it may well exist solely in her mind.

\paragraph{Hiding} The two protocol properties defined so far, implementing and committing, hold also in the traditional protocol (committing even holds there in a stronger, deterministic sense). We now define a new property---``hiding''---that is not present in the traditional protocol. In a nutshell: First, before player types are reported, players without exceedingly high computational power cannot distinguish between mechanisms better than if their prior, when conditioned on mechanism descriptions that are consistent with the announced properties (IR and DSIC, being $(\runrandlen,R)$-runnable, having a proof of length at most $L_p$ of this, and being of length at most $L_a$), were almost unchanged. Second, after the protocol concludes, players without exceedingly high computational power cannot distinguish between mechanisms better than if their prior, when conditioned on mechanism descriptions that are consistent with the announced properties until then (all of the above, and in addition what $M(t,\runrand)$ is), were almost unchanged.\footnote{While it might at first glance seem that making $\runrand$ public leaks too much information, this need not be the case. The mechanism designer can create the mechanism description by including, as part of the committed description, an arbitrary bijection of $\{0,1\}^{\runrandlen}$ onto itself to apply to $\runrand$ before calculating $M(t,\runrand)$. On the one hand, as discussed in \cref{re-permutation} on Page~\pageref{re-permutation}, this does not change the output distribution. On the other hand, with such a bijection applied, a given public value of $\runrand$ need not correspond to any particular structural part of the mechanism. Under a uniform relabeling of $\{0,1\}^{\runrandlen}$, observing $\runrand$ does not reveal more about the mechanism than if $\runrand$ were hidden.
}

\begin{definition}[Hiding Protocol\cryptofootnote{One may contrast our ``hiding'' desideratum with the standard, intricate cryptographic notion of secrecy through ``simulation.'' While cryptography many times shies away from non-simulation-based secrecy guarantees as in our ``hiding,'' the usual downfalls of non-simulation-based approaches are either circumvented here due to properties specific to the interactions in our protocols (e.g., the playbook not choosing the mechanism based on $\commitrand$) or are of minor interest in economics (e.g., a player participating in the mechanism only to save on computational power for some other, pathologically defined task). That said, all of our constructions in this paper also satisfy this stronger property, a full formulation of which for commit-and-hide protocols appeared in an earlier version of this paper. We omit this formulation here for brevity and to keep the exposition focused.}]\label{hiding}
A protocol $\bigl((\commitrandlen,\commitmsglen,\commitverify,\runrandlen,\runmsglen,\runverify),(\protcommitstrat,\protrunstrat)\bigr)$ is \emph{hiding} for a specification $(L_a,\runrandlen,R,L_p,B,\varepsilon)\in\Sigma$ if conditions~1 and~2 below hold for every distinguisher program~$\D$ computable in time at most~$B$, and every $A_1,A_2\in\mathcal{A}_{L_a,\runrandlen,R}$ and proofs $P_1,P_2\in\mathcal{P}_{L_p}$ such that each $P_i$ is a proof that $A_i$ is an $(\runrandlen,R)$-runnable description of an IR and DSIC mechanism.
\begin{enumerate}
    \item (Commitment step is hiding) $|p_1-p_2|\le\varepsilon$, where for every $i\in\{1,2\}$, $p_i$ is the probability---over the draw of commitment bits $\commitrand\sim U\bigl(\{0,1\}^{\commitrandlen}\bigr)$ and the randomness of $\protcommitstrat$---of the distinguisher $\D$ outputting ``This is $A_1$'' on input $(\commitrand,\commitmsg)$ corresponding to a run of the commitment phase (steps 0 through~1b) of the protocol where the mechanism designer plays the commitment playbook for $A_i,P_i$, i.e., $\protcommitstrat(A_i,P_i,\cdot)$.
    \item (Run step is hiding) For every type profile $t\in T$ and run bits $\runrand\in\{0,1\}^{\runrandlen}$, if~$M_{A_1}(t,\runrand)=M_{A_2}(t,\runrand)$ then $|q_1-q_2|\le\varepsilon$, where for every $i\in\{1,2\}$, $q_i$ is the probability---over the draw of commitment bits $\commitrand\sim U\bigl(\{0,1\}^{\commitrandlen}\bigr)$ and the randomness of $(\protcommitstrat,\protrunstrat)$---of the distinguisher $\D$ outputting ``This is $A_1$'' on input $(\commitrand,\commitmsg,t,\runrand,x,\runmsg)$ corresponding to a run of (all steps of) the protocol where the mechanism designer plays the playbook for $A_i,P_i$, i.e., $\bigl(\protcommitstrat(A_i,P_i,\cdot),\protrunstrat(A_i,\cdot,\cdot,\cdot,\cdot)\bigr)$.
\end{enumerate}
\end{definition}

\paragraph{Computationally unobtrusive}
Finally, we define what it means for a protocol to be ``computationally unobtrusive.'' In a nutshell, the computational requirements of ``well-behaved'' agents---those who participate in the protocol without malice, i.e., a mechanism designer who just runs the playbook and players who just run the verifiers---should be essentially the same as in the traditional protocol, allowing only an increase in these computational requirements that is smaller by many orders of magnitude than the computational power of adversaries against which the protocol protects.

As is common with computational requirements in computer science and statistics, our definition is asymptotic, and in our case, it is asymptotic in the parameters in $\sigma$. Specifically, we require that the rate of increase in the computational requirements of well-behaved agents, as we increase $B$ or decrease $\varepsilon$, be exceedingly modest, and otherwise the computational requirements should be asymptotically essentially the same as in the traditional protocol.

We formalize the first part (exceedingly modest increase as we increase $B$ or decrease $\varepsilon$) by requiring that the computational requirements of well-behaved agents be subpolynomial in $B$ and $\frac{1}{\varepsilon}$. A function is \emph{subpolynomial} if it is asymptotically smaller than (i.e., grows more slowly than) \emph{every} positive-degree power of its parameters.
E.g., $\log^7 x$ is subpolynomial in $x$, whereas $x^{1/3}$ and $x^{0.000001}$ are \emph{not} subpolynomial in $x$.\cryptofootnote{The subpolynomial dependence on $B$ and $\nicefrac{1}{\varepsilon}$ is a reformulation of the standard notion of cryptographic security, namely that the protocol runs in time polynomial in the security parameter, while making sure that an adversary whose computing power is \emph{any} polynomial in the security parameter succeeds with probability negligible in the security parameter. Our formulation starts from the adversary's computing power and attack-success probability (rather than from the security parameter) as the basic parameters, and hence a requirement that the running time in question be polynomial in the security parameter is equivalently restated as that running time being subpolynomial in these two basic parameters.}
We formalize the second part (asymptotically essentially the same as in the traditional protocol) by requiring that beyond the dependence on $B$ and $\varepsilon$, the computational requirements be asymptotically the same as in the traditional protocol, with at most additional factors that are logarithmic in parameters on which the computation in the traditional protocol depends linearly.

To formally define computational unobtrusiveness, we in fact define what it means for a \emph{family} of commit-and-run protocols to be computationally unobtrusive.\cryptofootnote{The standard approach in computer science is to define such a family implicitly using ``big-O notation,'' which with a large number of parameters might be challenging to navigate for readers who do not regularly encounter it (and even for those who do).} Here, by ``family of commit-and-run protocols'' we mean what we call a \emph{commit-and-run protocol catalog}, which is simply a protocol for each possible specification. We define this notion and extend the definitions of implementing, committing, and hiding onto it.

\begin{definition}[Commit-and-Run Protocol Catalog]
A \emph{commit-and-run protocol catalog} is a family of protocols
$
\bigl((\commitrandlen^\sigma,\commitmsglen^\sigma,\commitverify^\sigma,\runrandlen^\sigma,\runmsglen^\sigma,\runverify^\sigma), (\protcommitstrat^\sigma,\protrunstrat^\sigma)\bigr)_{\sigma\in\Sigma}
$
where for every specification $\sigma=(L_a,\runrandlen,R,L_p,B,\varepsilon)\in\Sigma$ we have that $\runrandlen^\sigma=\runrandlen$.\footnote{A short discussion on notation is in order here. One can think of a catalog as a machine whose ``control panel'' can be used to input a specification $\sigma$, and which outputs a protocol suitable for running mechanisms as specified by $\sigma$ with guarantees as specified by $\sigma$. Under this interpretation, $\runrandlen$, a component of $\sigma$, is a number typed into the control panel, which specifies how many run bits supported mechanisms may need. The value $\runrandlen^\sigma$, a component of the protocol corresponding to specification $\sigma$, is the number of random bits drawn in Step 3 of that protocol. (There is a semantic difference here: $\runrandlen$ specifies which mechanisms we should be able to support, while $\runrandlen^\sigma$ specifies how to mechanically run the protocol.)
The requirement $\runrandlen^\sigma=\runrandlen$ simply states that these two numbers should be one and the same: Drawing this number of random bits is both sufficient and necessary for running all mechanisms that we have to support.} Such a catalog is called:
\begin{itemize}
    \item 
\emph{implementing} if for every $\sigma\in\Sigma$, the protocol $\bigl((\commitrandlen^\sigma,\commitmsglen^\sigma,\commitverify^\sigma,\runrandlen^\sigma,\runmsglen^\sigma,\runverify^\sigma), (\protcommitstrat^\sigma,\protrunstrat^\sigma)\bigr)$ is implementing for~$\sigma$.
    \item 
\emph{committing} if for every $\sigma\in\Sigma$, the protocol $\bigl((\commitrandlen^\sigma,\commitmsglen^\sigma,\commitverify^\sigma,\runrandlen^\sigma,\runmsglen^\sigma,\runverify^\sigma), (\protcommitstrat^\sigma,\protrunstrat^\sigma)\bigr)$ is committing for $\sigma$.
    \item 
\emph{hiding} if for every $\sigma\in\Sigma$, the protocol $\bigl((\commitrandlen^\sigma,\commitmsglen^\sigma,\commitverify^\sigma,\runrandlen^\sigma,\runmsglen^\sigma,\runverify^\sigma), (\protcommitstrat^\sigma,\protrunstrat^\sigma)\bigr)$ is hiding for $\sigma$.
\end{itemize}

\end{definition}

Recall that in the traditional protocol, the total computational power required from the mechanism designer is linear in $L_a+L_p+R$: linear in $L_a+L_p$ in the commitment phase (sending the mechanism description $A$ and sending the proof $P$), and linear in $R$ in the run phase (running the mechanism). The computational power required from the players is also linear in $L_a+L_p+R$: linear in $L_a+L_p$ in the commitment phase (receiving the mechanism description and verifying the proof), and linear in $R$ in the run phase (verifying the mechanism outcome). As explained above, we seek asymptotically the same computational requirements from well-behaved agents in our protocol, up to subpolynomial factors in $B$ and $\frac{1}{\varepsilon}$ and logarithmic factors in $L_a$, $L_p$, and $R$.
 
\begin{definition}[Computationally Unobtrusive Protocol Catalog]\label{efficient}
A commit-and-run protocol catalog $\bigl((\commitrandlen^\sigma,\commitmsglen^\sigma,\commitverify^\sigma,\runrandlen^\sigma,\runmsglen^\sigma,\runverify^\sigma), (\protcommitstrat^\sigma,\protrunstrat^\sigma)\bigr)_{\sigma\in\Sigma}$
is \emph{computationally unobtrusive} if both of the following hold:
    \begin{itemize}
        \item
            The running time of (sampling from, i.e., computing an output distributed according to) the commitment playbook $\protcommitstrat^{\sigma}$ and the running time of the commitment verifier $\commitverify^{\sigma}$ are both at most $(L_a+L_p)\cdot\polylog(L_a+L_p)\cdot\subpoly(B,\frac{1}{\varepsilon})$, where $\polylog$ is a polylogarithmic function (i.e., a polynomial in the logarithms of its parameters) and $\subpoly$ is a subpolynomial function, and these two functions do not depend on $\sigma$.
        \item
            The running time of (sampling from) the run playbook $\protrunstrat^{\sigma}$ and the running time of the run verifier $\runverify^{\sigma}$ are both at most $R\cdot\polylog(R)\cdot\subpoly(B,\frac{1}{\varepsilon})$, where $\polylog$ is a polylogarithmic function and $\subpoly$ is a subpolynomial function, and these two functions do not depend on $\sigma$.
    \end{itemize}
\end{definition}

\section{General Existence Theorem}\label{general-existence}

In this \lcnamecref{general-existence}, we formally state our general theorem, establishing the existence of a commit-and-run protocol catalog that satisfies all our desiderata from \cref{desiderata}, allowing to run \emph{any mechanism} with the same desirable trust and secrecy guarantees as when run via a trustworthy mediator, yet without the need for any mediator or preexisting trust.

\begin{theorem}
\label{existence}
Under standard, widely believed computational infeasibility assumptions, there exists a hiding, committing, implementing, and computationally unobtrusive direct-revelation commit-and-run protocol catalog.
\end{theorem}

The statement ``Under standard, widely believed computational infeasibility assumptions'' in \cref{existence} should be interpreted along the lines of the discussion on cryptographic guarantees in the introduction: We prove \cref{existence} by showing that there exists an
implementing and computationally unobtrusive direct-revelation commit-and-run protocol catalog, as well as a family of widely studied computational problems, believed by expert cryptographers and computational complexity theorists to be hard, such that the following holds.\footnote{The hardness assumptions of these problems are often referred to as ``standard cryptographic assumptions'' in the cryptographic literature, and share a similar status to that of the computational assumption of $\mathrm{P}\ne\mathrm{NP}$ (despite technically being stronger): While they have neither been proven nor disproven, they have withstood extensive scrutiny by mathematicians and computer scientists, and their assumed hardness has been the basis for numerous theoretical and practical applications in various fields and industries.}\textsuperscript{,}\cryptofootnote{Our construction can, for example, be based upon any of the four following well-known computational infeasibility assumptions: (1) RSA, (2) LWE (learning with errors), (3) DDH (decisional Diffie--Hellman) and LPN (learning parity with noise), (4) Bilinear DDH\@. For more details, see \cryptocref{cp-proof-sec} in the supplementary material.} Any program that is either (1) a mechanism-designer strategy that breaks the committing guarantee of this protocol catalog, or (2) a distinguisher that breaks the hiding guarantee of this protocol catalog, can be directly used (in a precise, quantitative sense) to solve these hard problems, which are at the heart of a plethora of real-life cryptographic systems. This in particular means that any program that breaks the hiding or committing guarantees of \cref{existence} could be directly used to break into these real-life cryptographic systems.

We prove \cref{existence} following the high-level outline from the introduction (i.e., the commitment message contains a cryptographic commitment to a mechanism description and a zero-knowledge proof of properties of this description, and the run message contains a zero-knowledge proof of properly running the mechanism). To make the proof of \cref{existence} (and the gaps between our economics-centered desiderata and standard computer-science cryptographic guarantees that it navigates) as transparent as possible to readers without background in cryptography, in \cref{cp-sec} we restate the results from the cryptographic literature that we use in the proof by formulating a general, flexible ``one-stop-shop wrapper'' around all the existing cryptographic tools and definitions that are needed in the proof. In \cref{proofs}, in turn, we prove \cref{existence} utilizing the reformulated framework from \cref{cp-sec}. We hope that the self-contained nature and relatively lower barrier to entry of the reformulated, consolidated framework from \cref{cp-sec} might make it appealing also for independent use in other economic theory papers that may wish to avoid the learning curve associated with directly engaging with numerous papers from the cryptographic literature (as we do in \cryptocref{cp-proof-sec} in the supplementary material when we map out how the framework from \cref{cp-sec} follows from existing results).

\section{Discussion}\label{discussion}

In this paper, we introduce a general framework for committing to, proving properties of, and running any given mechanism, without disclosing it and without revealing, even in retrospect, more than is revealed by these properties and the outcome. Our framework decouples commitment from disclosure, decouples trust in the process from transparency of the process, and requires no third parties or mediators, trusted or otherwise, all without changing the strategy space of any player, practically maintaining strategic equivalence and equivalence in computational requirements to the traditional protocol in which the mechanism is publicly announced and then run. We provide several concrete, practical, ready-to-run examples as well as a general construction. In \cref{contracts} in the supplementary material, we also describe an extension of our framework for the realm of contracts.

The need for secrecy depends on the broader context in which the mechanism is used; in some contexts it might be crucial, while in some it may not be important. Hiding the mechanism might, however, be desirable even when secrecy is not a goal in and of itself. For example, it might be a means toward greater gains for the designer\footnote{\cite{MaskinT1990} show that a principal can gain by hiding her type if she knows the agent's prior over this type. While our framework does not require this informational assumption, it can facilitate the same gains when the assumption holds. Our framework furthermore eliminates the need for a trusted mediator that is embedded in informed-principal analyses \citep[see][who makes this assumption explicit]{Myerson1983} and keeps the principal's type hidden (beyond what is revealed by the realized outcome) even after running the mechanism.
}
or toward greater gains for all.\footnote{\cite{CopicP2008} study a bilateral-trade problem with incomplete information and show that bargaining through a neutral, mutually trusted mediator who keeps offers hidden might increase welfare. \cite{HornerV2009} consider a market-for-lemons bargaining setting with correlated values, and show that when keeping offers hidden until they are accepted, the process is more likely to reach agreement.
}
Hiding the mechanism might also be desirable from a behavioral perspective: \cite{GuillenH2018} find 
that advising subjects that the Top Trading Cycles matching mechanism is strategyproof, without describing the details of the mechanism, increases truthtelling rates in the mechanism compared to when the mechanism is also described.\footnote{\cite{DanzVW2022} show a different setting in which providing additional information contributes to strategic confusion and increases deviations from truthful reporting.} Should players learn to trust our machinery like they do browser encryption (as discussed in the introduction),
our framework might provide a path toward similar increases in truthtelling rates in various mechanism design settings even when participants do not \emph{a priori} trust the mechanism designer to run a mechanism with the described properties.\footnote{Non-truthtelling behavior has been observed when people interact with several common strategyproof mechanisms, including a sealed-bid second-price auction \citep{KagelL1993} and the Deferred Acceptance stable matching mechanism \citep{HassidimRS21,ShorrerS17}.}
In this sense, our framework could also be viewed as contributing to recent literature on how to best present mechanisms to players \citep{GonczarowskiHT2022,GonczarowskiHIT2024}.

A straightforward extension of our framework allows the revelation to each player of only the facets of the outcome that are payoff-relevant to her (e.g., what she gets and how much she pays, or where she is matched), revealing to each player even less than only the full outcome, yet maintaining credibility (cf.\ \citealp{AkbarpourL2020}).\footnote{This extension appeared in an earlier version of this paper. We omit it here for brevity and to keep the exposition focused.} In some settings, however, there might be benefit in revealing some non-payoff-relevant information to each player. For example, in a matching setting one might wish to publicly reveal some statistics about the mechanism's performance, such as how many students received their most-preferred institution.\footnote{This aggregate information is what, e.g., \cite{GonczarowskiKNR19} report when comparing their mechanism to alternative implementations.} In an auction setting, one might wish to publicly reveal whether there was a winner, or how many winners there were. Generally, there might be benefit in revealing both various properties of the mechanism before it is run (e.g., strategic or fairness properties) and non-payoff-relevant parts of the outcome after the mechanism is run. Our framework provides a fine lever through which arbitrary information of any of these kinds can be revealed or withheld.\footnote{This could even be combined---under certain additional conditions and limitations---with hiding the type reports; see \cref{hide-types} on Page~\pageref{hide-types}. Our framework could also be further extended to then selectively reveal only partial information about the type reports.} This establishes what could be called an \emph{un}revelation principle for mechanism design, paving the way for a further layer of \emph{revelation design} following the design of the rules of the mechanism within mechanism design.

In this paper, we focus on individual rationality and incentive compatibility as the properties of interest to be proven. However, that is merely an example, and our construction can easily work with any other properties satisfied by the mechanism, so long as a proof of reasonable length can be furnished for them.\footnote{This is only a technical constraint, as one would argue that any property known by the mechanism designer to hold would have such a proof; otherwise, how would the mechanism designer know that it holds?} For example, one could prove that a mechanism for bilateral trade is budget balanced, that a matching mechanism is stable, or that a voting mechanism has desirable properties such as various forms of fairness. One could even prove to regulators, without superfluously exposing trade secrets, that the mechanism abides by any relevant law, so long as this statement can be precisely formalized.\footnote{As part of an antitrust lawsuit, it has been claimed that Google tweaked reserve prices in auctions \citep{Nylen2023}. An ability of firms to prove that their mechanisms abide by various antitrust regulations without divulging trade secrets might have the potential to lower enforcement costs.
}
Sometimes, though, such statements cannot be formalized, whether because it is too hard or costly, or because one wishes to check whether the spirit, rather than the letter, of the law (or regulation) is upheld. Such applications, in which there might be no way around ``eyeballing'' the mechanism directly, do not fall within the scope of our framework. Another application we do not cover is choosing between mutually exclusive competing mechanisms or contracts offered by different principals, in which case just knowing that, e.g., each mechanism is IR and IC, or each contract satisfies IR and limited liability and incentivizes exertion of high effort, would not suffice. While cryptographic techniques could be used to evaluate a player's choice function over the set of proposed mechanisms/contracts while revealing only the final choice, formulating the details is outside the scope of this paper.

\clearpage
\appendix

\section[The Commit-then-Prove Scheme:\texorpdfstring{\\}{ }Reformulated, Consolidated Tools from Cryptographic Theory]{The Commit-then-Prove Scheme: Reformulated, Consolidated Tools from Cryptographic Theory}
\label[appendix]{cp-sec}

In this \lcnamecref{cp-sec}, we define what we call the ``commit-then-prove'' scheme. We both define the security requirements for this scheme and state its existence, which follows from existing results in the cryptographic literature. This scheme provides a general, flexible ``one-stop-shop wrapper'' around all the existing cryptographic tools and definitions that are needed in our proof of \cref{existence}. We hope that the self-contained nature and relatively lower barrier to entry of the commit-then-prove scheme might make this reformulated, consolidated tool appealing also for independent use in other economic theory papers that may wish to avoid the learning curve associated with directly engaging with numerous papers from the cryptographic literature (as we do in \cryptocref{cp-proof-sec} in the supplementary material when we map out how the results reformulated and consolidated in the current \lcnamecref{cp-sec} follow from existing results).

The ``commit-then-prove'' scheme is a set of programs for use by one \emph{prover} and one or more \emph{verifiers}. These programs are intended to be used in stages: In a \emph{commit} stage (there can be only one such stage), the prover uses the programs to irrevocably commit to some secret information $w$. In a \emph{prove} stage (there can be any number of such stages), the prover can use the programs to verifiably claim a statement $\phi$ about $w$, i.e., the verifiers can use the programs to verify that the prover (secretly) knows a proof $P$ such that $\ctprelation(w,\phi,P)$ holds, where $\ctprelation$ is a commonly known relation that verifies that $P$ is a formal proof of $\phi$ with respect to $w$. These requirements are formalized by way of emulating an \emph{ideal functionality}, called $\fctp$, that captures the expected correctness and secrecy properties of the scheme (see \cref{fig:fctp}). $\fctp$ is presented in terms of instructions to be followed by an imaginary ``trusted party.'' It should be stressed, though, that $\fctp$ is not meant to be followed by anyone; it is merely part of a mental experiment used to capture the security properties of the commit-then-prove scheme.
 
\pprotocol{\bf Functionality $\fctp^\ctprelation$}{The commit-then-prove functionality. This formulation is a modification of the Commit-and-Prove functionality from \cite{clos02}. While in \cite{clos02} the prover can deposit secret witnesses adaptively where each proof can relate to all witnesses deposited so far, for simplicity here proofs can relate only to the initial witness and the present one (denoted $P$); hence the name ``commit \emph{then} prove.'' The results in this \lcnamecref{cp-sec} mandate several additional properties beyond \cite{clos02}, including non-interactivity of the commit-then-prove functionality (i.e., each stage involving a single message generated by the prover), perfect completeness, uniformly distributed reference string, and identically distributed simulated reference string. See \cryptocref{cp-proof-sec} in the supplementary material for further details.
}{fig:fctp}{p}{

$\fctp^\ctprelation$
is
parametrized by a relation\footnote{\scriptsize$\ctprelation(w,\phi,P)$ should be interpreted as ``true if and only if $P$ is a (valid) formal proof that the statement $\phi$ holds with respect to $w$.'' In the context of our paper, for example, if $w$ is a mechanism description and $\phi$ is a formal statement of the form ``the description is an $(\runrandlen,R)$-runnable description of an IR and DSIC mechanism,'' then $\ctprelation(w,\phi,P)$ is true if and only if $P$ is a formal proof that $w$ is an $(\runrandlen,R)$-runnable description of an IR and DSIC mechanism. As another example, if $w$ is a mechanism description and $\phi$ is a statement of the form ``the described mechanism, when run on type profile $t$ using run bits $\runrand$, has outcome~$x$,'' then $\ctprelation(w,\phi,P)$ is true if and only if $P$ is a formal proof that $w$ describes a mechanism that when run on type profile $t$ using run bits $\runrand$, has outcome~$x$ (in this case, $P$ does nothing more than simply trace a run of $w$).} $\ctprelation(\cdot,\cdot,\cdot)$, a prover $C$, and the following four procedures:\footnote{\scriptsize In the $\fctp$-separation challenge, these procedures are constructed by the simulator~$\Sim$.}
\begin{itemize}
    \item \func{ChooseCommitmentToken}: takes no arguments, returns a string $c$ to be used as a commitment token.
    \item \func{ChooseProofToken}: takes $c,\phi,n$, returns a string $p$ to be used as a proof token.
    \item \func{ChooseCommittedInformation}: takes $c',\phi',n',p'$, returns $w'$ to be used in lieu of secret information for commitment $c'$.
    \item \func{ChooseProof}: takes $c',\phi',n',p'$ (but NOT $w'$), returns $P'$ with length $|P'|\le n'$ to be used in lieu of a proof of $\phi'$ with respect to $w'$, but may or may not be a valid proof that $\ctprelation(w',\phi',P')$ holds.
\end{itemize}

$\fctp^\ctprelation$ proceeds as follows (initially, no records exist):
\begin{enumerate}
\item {\bf Commit:}
Upon receiving input \msg{Commit$,w$} from $C$, where $w$ is a finite sequence of bits, do: If a \msg{Commit} message was previously received, then output \msg{Error} to $C$. Otherwise, let $c\leftarrow\func{ChooseCommitmentToken}()$. If a record of the form $(\cdot,c)$ already exists, then output \msg{Abort} to $C$.
Otherwise, record $(w,c)$ and output \msg{Commitment$,c$} to $C$.

\item {\bf Prove:}
Upon receiving input \msg{Prove$,\phi,n,P$} from $C$,\footnote{\scriptsize One should think of $n$ as a revealed upper bound on the length of the proof $P$.} do: 
If no \msg{Commit} message has yet been received, then output \msg{Error} to $C$. Otherwise, let $w$ be what was received as part of the \msg{Commit} message and let $c$ be what was returned. If $\ctprelation(w,\phi,P)$ does not hold, then output \msg{Proof,False} to $C$. Otherwise, let $p\leftarrow\func{ChooseProofToken}(c,\phi,n)$. If a record of the form $(c,\phi,n,p,\cdot)$ already exists, then output \msg{Abort} to $C$. Record $(c,\phi,n,p,\True)$ and output \msg{Proof$,p$} to $C$.

\item {\bf Verify:}
Upon receiving input \msg{Verify$,c',\phi',n',p'$} from any party,
\begin{itemize} 
\item If a record $(c'',\phi'',n'',p'',x)$ exists
where $(c'',\phi'',n'', p'')=(c',\phi',n',p')$, 
then output \msg{Verified$,x$}.
\item
If no record of the form $(\cdot,c')$ exists, then choose $w'\leftarrow\func{ChooseCommittedInformation}(c',\phi',n',p')$ and record $(w',c')$. Otherwise,~let $w'$ be such that the record $(w',c')$ exists.
\item
Let $P'\leftarrow\func{ChooseProof}(c',\phi',n',p')$. If $\ctprelation(w',\phi',P')$ holds, then let $x=\True$; otherwise, let $x=\False$. Record $(c',\phi',n',p',x)$, and output \msg{Verified$,x$}.
\end{itemize}
\end{enumerate}
}

Each of the programs in the commit-then-prove scheme is parametrized by a \emph{security parameter}. To streamline the use of the commit-then-prove scheme in our context, we restrict ourselves to programs whose output length is always the same and depends only on the security parameter. We call such a program \emph{output-regular}.

\begin{definition}[Commit-then-Prove Scheme]
\label{def:ctp}
A \emph{commit-then-prove scheme} is a quadruple of output-regular programs $\pi=(\mathit{drawrefstring},\mathit{commit},\mathit{prove},\mathit{verify})$, parametrized by a \emph{security parameter} $\lambda$. A commit-then-prove scheme $\pi$ is \emph{secure} for relation~$\ctprelation$, maximum committed information length $n_w$, maximum statement length $n_\phi$, and maximum proof length $n_P$ if two requirements are met. The first requirement is that there exist:
\begin{enumerate}
\item a \emph{parameter translation function} $\tau:\mathbb{N}\times\mathbb{R}_{>0}\rightarrow\mathbb{N}$, where $\lambda=\tau(B,\varepsilon)$ is subpolynomial in $B$ and $\nicefrac{1}{\varepsilon}$;
    \item a \emph{simulator program} $\Sim$ that, given a security parameter $\lambda$ and a sequence of bits $\simrand\in\{0,1\}^{\simrandlen}$ that we think of as an internal source of randomness, generates the following five procedures:
    \begin{itemize}
    \item \func{SimulateRefString}: takes no arguments, returns a binary string $s_{\text{sim}}$ to be used as a simulated reference string (see below). We require that if $\simrand$ is drawn uniformly at random, then the distribution of $s_{\text{sim}}$ is identical to that of the output of $\mathit{drawrefstring}(\lambda)$.
    \item \func{ChooseCommitmentToken}, \func{ChooseProofToken}, \func{ChooseCommittedInformation}, \linebreak
    \func{ChooseProof}: to be used by $\fctp$ as described momentarily.
\end{itemize} 
\end{enumerate}
The combined running time of $\Sim$ and each one of these procedures must be polynomial in $\lambda$, in the maximum committed information length $n_w$, in the maximum statement length $n_\phi$, and in the maximum proof length $n_P$. 

To define the second requirement for $\pi$ to be secure---intuitively, that its four programs closely emulate the ideal functionality $\fctp$---we first define the following ``challenge'' for an adversary program $\A$, which is officiated by an impartial \emph{challenge host}.

\paragraph{\underline{The $\fctp$-Separation Challenge:}}
\begin{enumerate}
\item The challenge host chooses a bit $b$ uniformly at random.
\item If $b=0$ then the challenge host runs $\A$ in an interaction with the scheme $\pi=(\mathit{drawrefstring},\mathit{commit},\mathit{prove},\mathit{verify})$, parametrized by $\lambda$. That is:
\begin{enumerate}
\item The challenge host computes\footnote{$\mathit{drawrefstring}$ is a randomized program, as are $\mathit{commit}$ and $\mathit{prove}$.} a \emph{reference string} $s\sim \mathit{drawrefstring}(\lambda)$ and sends $s$ to $\A$.
    \item $\A$ can repeatedly choose between one of the following three options.
    \begin{enumerate}
        \item $\A$ sends \msg{Commit$,w$} to the challenge host (where $w$ is of length at most $n_w$). In this case, if a \msg{Commit} message was already previously sent to the challenge host, the challenge host sends \msg{Error} to $\A$. Otherwise, the challenge host computes $(c,\mu)\sim\mathit{commit}(\lambda,w,s)$ and sends \msg{Commitment$,c$} to $\A$. (And remembers the \emph{state} $\mu$.)
        \item
        $\A$ sends \msg{Prove$,\phi,n,P$} to the challenge host, where $n\le n_P$, and where $\phi$ and $P$ are of respective length at most $n_\phi$ and $n$. 
        In this case, if a \msg{Commit} message was never previously sent to the challenge host, the challenge host sends \msg{Error} to $\A$. Otherwise, the challenge host computes $(p,\mu')\sim\mathit{prove}(\lambda,\mu,\phi,n,P)$, updates the state $\mu\leftarrow\mu'$, and sends \msg{Proof$,p$} to $\A$.
        \item $\A$ sends \msg{Verify$,c',\phi',n',p'$} to the challenge host, where $n'\le n_P$, where $\phi'$ is of length at most $n_\phi$, and where $c',p'$ are of the length of the output of the (output-regular) programs in $\pi$.
        In this case, the challenge host computes $x=\mathit{verify}(\lambda,s,c',\phi',n',p')$ and sends \msg{Verified$,x$} to~$\A$.
    \end{enumerate}
\end{enumerate}
\item If $b=1$ then the challenge host runs $\A$ in an interaction with $\fctp^\ctprelation$ (see 
\cref{fig:fctp}) and $\Sim$ with security parameter $\lambda$.
That is:
\begin{enumerate}
\item \begin{sloppypar}
The challenge host draws $\simrand\in\{0,1\}^{\simrandlen}$ uniformly at random, and runs $\Sim=\Sim(\simrand,\lambda)$ to generate the five procedures \func{SimulateRefString}, \func{ChooseCommitmentToken}, \func{ChooseProofToken}, \func{ChooseCommittedInformation}, and \func{ChooseProof}.
Next the challenge host invokes an instance of $\fctp^\ctprelation$ parametrized with \func{ChooseCommitmentToken}, \func{ChooseProofToken}, \func{ChooseCommittedInformation}, and \func{ChooseProof}. The challenge host then runs \func{SimulateRefString} to obtain a simulated reference string $s_{\text{sim}}$, which the challenge host sends to $\A$ as a reference string.
\end{sloppypar}
\item $\A$ can repeatedly choose between one of the following three options: to send \msg{Commit$,w$} to the challenge host, to send \msg{Prove$,\phi,n,P$} to the challenge host, or to send \msg{Verify$,c',\phi',n',p'$} to the challenge host. Either way, the challenge host forwards the message to the instance of $\fctp^\ctprelation$ initiated in the previous step,
and sends the returned output to $\A$.
\end{enumerate}
\item Finally, $\A$ outputs a value $b'$. We say that $\A$ \emph{wins the challenge} if $b'=b$.
\end{enumerate}
An adversary program $\A$ has \emph{separation advantage} $\varepsilon$ if it wins the $\fctp$-separation challenge with probability $\nicefrac{1}{2}+\varepsilon$ (i.e., it has an advantage of $\varepsilon$ compared to a uniform guess). This captures the idea that from the adversary's lens, the scheme closely emulates $\fctp$.

The second requirement for a commit-then-prove scheme $\pi$ to be \emph{secure} is that for every running time bound $B\in\mathbb{N}$ and every separation advantage bound $\varepsilon>0$, it holds that every adversary program $\A$ that runs in time at most $B\cdot C_\pi(\lambda)$---where $\lambda=\tau(B,\varepsilon)$ and $C_\pi(\lambda)$ is the sum of the maximum running times of $\mathit{commit}$, $\mathit{prove}$, and $\mathit{verify}$ with security parameter $\lambda$---has separation advantage at most $\varepsilon$.
Such an adversary $\A$ that has separation advantage greater than $\varepsilon$ is said to \emph{break the security} of $\pi$.\cryptofootnote{The requirement that $\pi$ is a secure commit-then-prove scheme is a restatement, with three additions discussed in \cryptocref{cp-proof-sec} in the supplementary material, of the requirement that $\pi$ UC-realizes $\fctp^\ctprelation$ as in \cite{Canetti20}, which we managed to simplify by avoiding the definition of the UC execution framework in its full generality before specializing it for $\fctp^\ctprelation$.}
\end{definition}

\begin{definition}[Perfect Completeness]\label{def:ctp-complete}
An adversary $\A$ is \emph{benign} if in all of its \linebreak \msg{Verify$,c,\phi,n,p$} messages the value $c$ is the value returned by the challenge host in response to an earlier \msg{Commit$,w$} message and the value $p$ is the value returned by the challenge host in response to an earlier \msg{Prove$,\phi,n,P$} message, and if it does not otherwise use in the calculation of its output any values returned by the challenge host in response to \msg{Commit$,\cdot$} or \msg{Prove$,\cdot$} messages. A commit-then-prove scheme has \emph{perfect completeness} if there exists a simulator $\Sim_{\text{benign}}$ such that all benign adversaries have separation advantage $\varepsilon=0$ in the $\fctp$-separation challenge with simulator $\Sim_{\text{benign}}$.
\end{definition}

\begin{sloppypar}
\begin{definition}[Uniform Reference String]
A commit-then-prove scheme $\pi=(\mathit{drawrefstring},\mathit{commit},\mathit{prove},\mathit{verify})$ has a \emph{uniform reference string} if $\mathit{drawrefstring}$ simply draws a uniform random string of some predefined length (that may depend on $\lambda$).
\end{definition}
\end{sloppypar}

\begin{definition}[Feasible Computability]
A commit-then-prove scheme $\pi=\linebreak(\mathit{drawrefstring},\mathit{commit},\mathit{prove},\mathit{verify})$ is \emph{feasibly computable} if all of the following hold:
\begin{itemize}
\item
The running time of (sampling from) $\mathit{drawrefstring}$ is polynomial in the security parameter~$\lambda$.
\item
The running time of (sampling from) $\mathit{commit}$ is quasilinear in the maximum committed information length $n_w$ and polynomial in the security parameter~$\lambda$. That is, the running time is at most $n_w\cdot\polylog(n_w)\cdot\poly(\lambda)$.
\item
The running times of (sampling from) $\mathit{prove}$ and $\mathit{verify}$ are quasilinear in the maximum committed information length $n_w$, the length of the statement $\phi$, the value $n$, and the maximum running time of $\ctprelation(w',\phi,P')$ over all $w',P'$ with respective lengths at most $n_w$ and $n$, and polynomial in the security parameter~$\lambda$. That is, the running times are at most $(n_w+|\phi|+n+R_{\ctprelation(\cdot,\phi,\cdot)})\cdot\polylog(n_w+|\phi|+n+R_{\ctprelation(\cdot,\phi,\cdot)})\cdot\poly(\lambda)$, where $|s|$ denotes the length of the string $s$ and $R_{\ctprelation(\cdot,\phi,\cdot)}$ denotes the maximum running time of $\ctprelation(w',\phi,P')$ over all $w',P'$ with respective lengths at most $n_w$ and $n$.
\end{itemize}
\end{definition}

We now state the main result of this \lcnamecref{cp-sec}, which establishes the existence of a commit-then-prove scheme with all of the above properties under suitable computational infeasibility assumptions. 

\begin{theorem}\label{ctp}
Under standard, widely believed computational infeasibility assumptions, there exists a feasibly computable secure commit-then-prove scheme with perfect completeness and uniform reference string for any computable relation $\ctprelation$.
\end{theorem}

See \cryptocref{cp-proof-sec} in the supplementary material for how \cref{ctp} immediately follows from existing results in the cryptographic literature. We note that these results are constructive. That is, they provide an algorithmic way for transforming any adversary that breaks the security of the commit-then-prove scheme into an adversary that breaks the security of one of the underlying standard computational infeasibility assumptions.

\section[Proof of Theorem~\ref{existence} Based on the Commit-then-Prove Scheme]{Proof of Theorem~\ref{existence} Based on the \texorpdfstring{\\}{ }Commit-then-Prove Scheme}\label[appendix]{proofs}

In this \lcnamecref{proofs}, we prove \cref{existence} based on the framework from \cref{cp-sec}, constructing the commit-and-run protocol catalogs guaranteed to exist by this \lcnamecref{existence}.

\begin{proof}[Proof of \cref{existence}]
To prove \cref{existence}, we construct a commit-and-run protocol \linebreak $\bigl((\commitrandlen^\sigma,\commitmsglen^\sigma,\commitverify^\sigma,\runrandlen^\sigma,\runmsglen^\sigma,\runverify^\sigma), (\protcommitstrat^\sigma,\protrunstrat^\sigma)\bigr)$ for each specification $\sigma\in\Sigma$. Let $\sigma=(L_a,\runrandlen,R,L_p,B,\varepsilon)$. We first define a commit-then-prove scheme $\pi^\sigma=(\mathit{drawrefstring},\mathit{commit},\mathit{prove},\mathit{verify})$ using \cref{ctp}, and then use this scheme to construct the commit-and-run protocol.

\paragraph{Defining an appropriate commit-then-prove scheme} Let $\pi^\sigma=(\mathit{drawrefstring},\linebreak\mathit{commit},\mathit{prove},\mathit{verify})$ be a commit-then-prove scheme as guaranteed by \cref{ctp} for the following parameters, which are defined using the specification $\sigma=(L_a,\runrandlen,R,L_p,B,\varepsilon)$:
\begin{itemize}
    \item $\ctprelation(A,\phi,P)$ is the relation ``$P$ is a proof that statement $\phi$ holds for $A$.''
    \item Set $n_w \leftarrow L_a$.
    \item Set $n_P \leftarrow \max\{L_p,c\cdot R\}$, where $c$ is a constant such that $c\cdot R$ upper bounds the maximum length of a proof, for a program that runs in time at most $R$, of what its output is for a specific given input. (By assumption, this length is proportional to $R$.)
    \item Set $n_\phi \leftarrow $ the maximum length of any formal statement that formalizes a statement of one of the following two forms:
    \begin{itemize}
        \item ``This is an $(\runrandlen,R)$-runnable description of an IR and DSIC mechanism.''
        \item ``This description, when run on type profile $t$ using run bits~$\runrand$, has outcome~$x$.''
    \end{itemize}
    \item Set $\lambda \leftarrow \tau\bigl(B+2,\nicefrac{\varepsilon}{4}\bigr)$.
\end{itemize}

\paragraph{Defining the grammar rules} We define the grammar rules $(\commitrandlen^\sigma,\commitmsglen^\sigma,\commitverify^\sigma,\runrandlen^\sigma,\runmsglen^\sigma,\runverify^\sigma)$ for the commit-and-run protocol using the specification $\sigma=(L_a,\runrandlen,R,L_p,B,\varepsilon)$ and the scheme $\pi^\sigma=(\mathit{drawrefstring},\mathit{commit},\mathit{prove},\mathit{verify})$ as follows:
\begin{itemize}
\item $\commitrandlen^\sigma$ --- number of bits returned by a call to $\mathit{drawrefstring}$.
\item $\commitmsglen^\sigma$ --- sum of number of bits returned by $\mathit{commit}$ and number of bits returned by $\mathit{prove}$. We henceforth use the notation $\commitmsg=(C,p_\commitsubs)$.
\item $\commitverify^\sigma(\commitrand,\commitmsg)=\commitverify^\sigma\bigl(\commitrand,(C,p_\commitsubs)\bigr)$ is computed by running $\mathit{verify}(\lambda,\commitrand,C,\commitcorrectprop,L_p,p_\commitsubs)$, where $\commitcorrectprop$ is the formal statement that formalizes the claim ``this is an $(\runrandlen,R)$-runnable description of an IR and DSIC mechanism.''
\item ($\runrandlen^\sigma$ is copied from $\sigma$.)
\item $\runmsglen^\sigma$ --- number of bits returned by $\mathit{prove}$.
\item\begin{sloppypar} $\runverify^\sigma\bigl(\commitrand,\commitmsg,t,\runrand,(x,p_\runsubs)\bigr)=\runverify^\sigma\bigl(\commitrand,(C,p_\commitsubs),t,\runrand,(x,p_\runsubs)\bigr)$ is computed by running $\mathit{verify}\bigl(\lambda,\commitrand,C,\runcorrectprop(t,\runrand,x),c\cdot R,p_\runsubs\bigr)$, where $\runcorrectprop(t,\runrand,x)$ is the formal statement that formalizes the claim ``this description, when run on type profile~$t$ using run bits $\runrand$, has outcome~$x$'' and $c$ is as in the definition of the value of $n_P$ earlier in this proof.
\end{sloppypar}
\end{itemize}

\paragraph{Defining the playbook} We define the playbook $(\protcommitstrat^\sigma,\protrunstrat^\sigma)$ as follows:
\begin{itemize}
    \item Commitment strategy for mechanism designer with mechanism description~$A$ and proof $P$ that $A$ satisfies $\commitcorrectprop$:
    \begin{itemize}
    \item Compute $(C,s')\sim\mathit{commit}(\lambda,A,\commitrand)$.
    \item
    Compute $(p_\commitsubs,s)\sim\mathit{prove}(\lambda,s',\commitcorrectprop,L_p,P)$.
    \item
    Let $\commitmsg\leftarrow(C,p_\commitsubs)$ and return $(\commitmsg,s)$.
    \end{itemize}
    \item
    Run strategy for mechanism designer with mechanism description $A$:
    \begin{itemize}
    \item
    Compute $x\leftarrow M_A(t,\runrand)$.
    \item
    Let $P_\runsubs$ be a proof that $A$ satisfies $\runcorrectprop(t,\runrand,x)$ (i.e., a proof that follows the steps of the calculation of $M_A(t,\runrand)$ and proves that each step is correct).
    \item
    Compute $(p_\runsubs,s'')\sim\mathit{prove}\bigl(\lambda,s,\runcorrectprop(t,\runrand,x),c\cdot R,P_\runsubs\bigr)$, where $c$ is as in the definition of the value of $n_P$ earlier in this proof.
    \item
    Let $\runmsg\leftarrow p_\runsubs$ and return $(x,\runmsg)$.
    \end{itemize}
\end{itemize}

\begin{sloppypar}
\paragraph{Properties of the protocol catalog} We have concluded the definition of our protocol for each specification $\sigma$. It remains to show that the resulting protocol catalog $\bigl((\commitrandlen^\sigma,\commitmsglen^\sigma,\commitverify^\sigma,\runrandlen^\sigma,\runmsglen^\sigma,\runverify^\sigma), (\protcommitstrat^\sigma,\protrunstrat^\sigma)\bigr)_{\sigma\in\Sigma}$ is implementing, committing, hiding, and computationally unobtrusive.
\end{sloppypar}

\paragraph{Computationally unobtrusive} The catalog is computationally unobtrusive (\cref{efficient}) because the commit-then-prove schemes $\pi^\sigma$ are feasibly computable by \cref{ctp}, and since $R\ge L_a$ and the running time of $\ctprelation$ is linear in the sum of the lengths of its arguments.

\paragraph{Implementing} To prove that the catalog is implementing (\cref{implementing}), for every $\sigma$ note that since the playbook by definition calculates the outcome $M_A(t,\runrand)$, it suffices to prove that verifications always pass when the mechanism designer follows the playbook. We will show that this follows from $\pi^\sigma$ having perfect completeness. Recall that by \cref{def:ctp-complete} there exists a simulator $\Sim_{\text{benign}}$ such that the separation advantage of every benign adversary $\A$ is $\varepsilon=0$.

Assume for contradiction that for some specification $\sigma$, verifications fail with positive probability for some $A$, $P$, $t$, and $\runrand$. Construct a benign adversary $\A$ for the $\fctp$-separation challenge as follows. $\A$ simply runs our commit-and-run protocol, playing the roles of both the mechanism designer (using the playbook) and the players according to $A$, $P$, $t$, and $\runrand$, replacing calls to $\mathit{commit}$, $\mathit{prove}$, and $\mathit{verify}$ with respective messages to the challenge host. $\A$ outputs $b'=1$ if verifications pass, and $b'=0$ if any verification fails.

We obtain a contradiction by observing that this benign adversary has positive separation advantage. Indeed, note that when $b=1$ (interacting with $\Sim_{\text{benign}}$), verifications always pass, while when $b=0$ (interacting with $\pi^\sigma$), by our assumption they fail with positive probability.

\paragraph{Committing} To prove that the catalog is committing (\cref{committing}), for every specification $\sigma$ and mechanism-designer strategy $(\commitstrat,\runstrat)$ computable in time at most $B$ we must show that there exists an appropriate (not necessarily feasibly computable\cryptofootnote{Recall that our notion of committing differs from the standard cryptographic notion of extractability, which requires the ability to feasibly approximately sample pairs $(\commitrand,A)$. We have chosen our notion to maximize economic usability. We nonetheless note that the constructions in this paper satisfy both notions.}) function $\alpha:\{0,1\}^{\commitrandlen}\to\Delta(\mathcal{A}_{L_a,\runrandlen,R})$. We implicitly define this function by directly defining the coupling $\Xi$ (which itself need not necessarily be feasibly computable) of the mechanism-designer commitment strategy and this function, using the simulator~$\Sim$ that exists for $\pi^\sigma$ by \cref{def:ctp}:
\begin{enumerate}
\item 
\begin{sloppypar}
Draw $\simrand\in\{0,1\}^{\simrandlen}$ uniformly at random from all sequences $\simrand$ such that the procedure \func{SimulateRefString} constructed by $\Sim(\simrand,\lambda)$ chooses~$\commitrand$ as the simulated reference string.
Initialize $\Sim=\Sim(\simrand,\lambda)$.
\end{sloppypar}
\item Let $(\commitmsg,s)\sim\commitstrat(\commitrand)$.
\item If not $\commitverify(\commitrand,\commitmsg)$, then return $\bigl((\commitmsg,s),A\bigr)$ for an arbitrary description $A\in \mathcal{A}_{L_a,\runrandlen,R}$ of an IR and DSIC mechanism.
\item Let $(c,p_\commitsubs)\leftarrow \commitmsg$.
\item\begin{sloppypar} Recall that the simulator $\Sim$, together with constructing the procedure $\func{SimulateRefString}$ that chooses the simulated reference string (which is here used as $\commitrand$), constructs four procedures, including \func{ChooseCommittedInformation} and \func{ChooseProof}; let $A\leftarrow\func{ChooseCommittedInformation}(c,\commitcorrectprop,L_p,p_\commitsubs)$.
\end{sloppypar}
\item Let $P\leftarrow\func{ChooseProof}(c,\commitcorrectprop,L_p,p_\commitsubs)$.
\item If $\ctprelation(A,\commitcorrectprop,P)$ does not hold, then let $A$ instead be an arbitrary description in $\mathcal{A}_{L_a,\runrandlen,R}$ of an IR and DSIC mechanism.
\item
Return $\bigl(((c,p_\commitsubs),s),A\bigr)$.
\end{enumerate}

By definition of $\Xi$, its first marginal is distributed like $\commitstrat$. By definition of $\Xi$, $\ctprelation$, and $\commitcorrectprop$, any description returned with positive probability by the second marginal of $\Xi$ is of an IR and DSIC mechanism.

Let $t\in T$ and $\runrand\in\{0,1\}^{\runrandlen}$. It remains to show that:
\begin{multline}
\label{commit-violation}
    \Pr_{\commitrand\sim U(\{0,1\}^{\commitrandlen}),\Xi,\runstrat}\bigl[\commitverify\bigl(\commitrand,\commitmsg\bigr)\ \&\ \runverify\bigl(\commitrand,\commitmsg,t,\runrand,(x,\runmsg)\bigr)\ \&\ x\ne M_A(t,\runrand)\bigr]\le \varepsilon \\
\text{for $\bigl((\commitmsg,s),A\bigr)\sim\Xi(\commitrand)$ and $(x,\runmsg)\sim \runstrat\bigl(\commitrand,(\commitmsg,s),t,\runrand\bigr)$.}
\end{multline}

Let $q$ be the probability on the left-hand side of \cref{commit-violation}.
To prove that $q\le\varepsilon$ as required, we construct an adversary $\A$ for the $\fctp$-separation challenge as follows:
\begin{itemize}
    \item Run our commit-and-run protocol with the reference string of the $\fctp$-separation challenge as $\commitrand$ using the mechanism-designer strategy $(\commitstrat,\runstrat)$, with player reports $t$ and run bits $\runrand$. Whenever there is a call to $\mathit{verify}$, in addition to this call also send a respective message to the challenge host.
    \item
    If the verifications pass (i.e., both $\commitverify$ and $\runverify$ return $\True$) and yet the challenge host returns $\False$ for a \msg{Verify} message (either the one in the commitment phase or the one in the run phase), then output $b'=1$. Otherwise, output $b'=0$.
\end{itemize}
The running time of $\A$ is at most $B+2C_{\pi^\sigma}(\lambda)$ (where $C_{\pi^\sigma}(\lambda)$ is as in \cref{def:ctp}). Therefore, by the guarantee in \cref{def:ctp}, its separation advantage is at most $\nicefrac{\varepsilon}{4}$.

Recall that if $\simrand$ is chosen uniformly at random, then $\func{SimulateRefString}$ draws its output reference string according to the true reference string distribution, i.e., uniformly at random since $\pi^\sigma$ has a uniform reference string. Therefore, $\simrand$ used by $\Xi$ in \cref{commit-violation} is distributed uniformly at random in $\{0,1\}^{\simrandlen}$, just like in the $\fctp$-separation challenge. This creates a natural coupling between the probability spaces in \cref{commit-violation} and in the $\fctp$-separation challenge when the interaction is with $\fctp^\ctprelation$ and $\Sim$ (i.e., when $b=1$).

Note that $\A$ never sends any \msg{Commit} or \msg{Prove} messages to the challenge host. Therefore, if $\A$ is run in an interaction with $\fctp^\ctprelation$ and $\Sim$ (i.e., if $b=1$), upon receipt of the commitment-verification \msg{Verify} message, $\fctp$ runs \func{ChooseCommittedInformation} with inputs $c$, $\commitcorrectprop$, $L_p$, and $p_\commitsubs$ to obtain a (valid or invalid) description $A$, and \func{ChooseProof} with inputs $c$, $\commitcorrectprop$, $L_p$, and $p_\commitsubs$ to obtain a (valid or invalid) proof $P$, and returns $\True$ only if $\ctprelation(A,\commitcorrectprop,P)$ holds; moreover, upon receipt of the run-verification \msg{Verify} message, $\fctp$ returns $\True$ only if $\ctprelation\bigl(A,\runcorrectprop(t,\runrand,x),P_\runsubs\bigr)$ holds for some~$P_\runsubs$. Therefore, the condition under which $\A$ outputs $b'=1$ holds whenever the event whose probability is on the left-hand side of \cref{commit-violation} holds. On the other hand, if $\A$ is run in interaction with $\pi^\sigma$ (i.e., $b=0$) then if any verification passes, then the challenge host returns $\True$ for the corresponding \msg{Verify} message, and so the condition under which $\A$ outputs $b'=1$ never holds. To sum up, with probability at least $q$, we have that $\A$ wins the $\fctp$-separation challenge. Otherwise $\A$ wins with probability~$\nicefrac{1}{2}$. Overall, $\A$ wins with probability at least $q\cdot1+(1-q)\cdot\nicefrac{1}{2}=\nicefrac{1}{2}+\nicefrac{q}{2}$, i.e., has separation advantage at least $\nicefrac{q}{2}$. Hence, $q\le\varepsilon$, as required.\footnote{Recall that otherwise, $\A$ breaks the security of $\pi^\sigma$, and hence can be feasibly transformed into an adversary that breaks the security of the underlying standard, widely believed computational infeasibility assumption.}

\paragraph{Hiding} To prove that the catalog is hiding (\cref{hiding}), assume for contradiction that there exists a distinguisher program $\D$ that demonstrates that the protocol is not hiding for some specification $\sigma$ by distinguishing between two description-proof pairs $(A_1,P_1),(A_2,P_2)$ with probability difference greater than $\varepsilon$.

We first consider the case in which $\D$ demonstrates that the run step is not hiding. When this is the case, there exist $t\in T$ and $\runrand\in\{0,1\}^{\runrandlen}$ such that $|q_1-q_2|>\varepsilon$.

We construct two adversaries $\A_1,\A_2$ for the $\fctp$-separation challenge as follows. For every $i\in\{1,2\}$, the adversary $\A_i$ simply runs our commit-and-run protocol, without any of the verification steps, playing the roles of both the mechanism designer (using the playbook) and the players according to $A=A_i$, $P=P_i$, $t$, $\runrand$, where $\commitrand$ is taken to equal the reference string, replacing calls to $\mathit{commit}$ and $\mathit{prove}$ (no calls to $\mathit{verify}$ occur as no verification steps are run) with respective messages to the challenge host. $\A_i$ then runs $\D$ on the resulting output, outputting $b'=0$ if $\D$ outputs ``This is $A_1$'' and outputting $b'=1$ otherwise.

We note that since $M_{A_1}(t,\runrand)=M_{A_2}(t,\runrand)$, running the $\fctp$-separation challenge with $b=1$ with either $\A_1$ or $\A_2$ results in the exact same (distribution over) input to $\D$. Therefore, the probability of each of these adversaries outputting $b'=0$ when $b=1$ is the same; denote this probability by $p$. Since $|q_1-q_2|>\varepsilon$, by the triangle inequality there exists $i\in\{1,2\}$ such that $|q_i-p|>\nicefrac{\varepsilon}{2}$. Denote $\A=\A_i$.

The running time of $\A$ is at most $2C_{\pi^\sigma}(\lambda)+B$ (where $C_{\pi^\sigma}(\lambda)$ is as in \cref{def:ctp}; note that since $t$ and $\runrand$ are fixed when we construct the adversaries, the outcome $x=M_{A_i}(t,\runrand)$ and the proof $P_\runsubs$ that $A_i$ satisfies $\runcorrectprop(t,\runrand,x)$ can be precomputed and hardcoded into $\A_i$, avoiding an additional running time of $R$). We obtain a contradiction by observing that $\A$ has separation advantage greater than~$\nicefrac{\varepsilon}{4}$.\footnote{Recall that this would mean that $\A$ breaks the security of $\pi^\sigma$, and hence can be feasibly transformed into an adversary that breaks the security of the underlying standard, widely believed computational infeasibility assumption.} For $x,y\in\{0,1\}$, denote by $\Pr[b'=x\mid b=y]$ the probability of $\A$ outputting $b'=x$ when truly $b=y$. By definition of $\D$ and construction of $\A$, we have that $\Pr[b'=0\mid b=1]=p$ and $\Pr[b'=0\mid{b=0}]=q_i$. Hence, $\bigl|\Pr[b'=0\mid b=0]-\Pr[b'=0\mid b=1]\bigr|>\nicefrac{\varepsilon}{2}$. Without loss of generality, assume that $\Pr[b'=0\mid b=0]>\Pr[b'=0\mid b=1]+\nicefrac{\varepsilon}{2}$ (otherwise switch the two outputs of the adversary $\A$). We have:
\begin{multline*}
\Pr[\text{$\A$ wins}]=\Pr[b'=0~\&~b=0]+\Pr[b'=1~\&~b=1]=\\\Pr[b'=0\mid b=0]\cdot\nicefrac{1}{2}+\bigl(1-\Pr[b'=0\mid b=1]\bigr)\cdot\nicefrac{1}{2}=\\\nicefrac{1}{2}+\nicefrac{1}{2}\cdot\bigl(\Pr[b'=0\mid b=0]-\Pr[b'=0\mid b=1]\bigr)>\nicefrac{1}{2}+\nicefrac{\varepsilon}{4},
\end{multline*}
as required.

The case in which $\D$ demonstrates that the commitment step is not hiding is handled similarly, having each $\A_i$ only run our commit-and-run protocol until the end of the commitment step.
\end{proof}

{
\small
\singlespacing
\phantomsection
\addcontentsline{toc}{section}{References}
\putbib

}

\clearpage 

\end{bibunit}

\addtocontents{toc}{\protect\setcounter{tocdepth}{1}}

\addtocontents{toc}{\vspace{1em}}
\setcounter{page}{1}
\renewcommand{\thepage}{S.\arabic{page}}
\setcounter{table}{0} 
\renewcommand{\thetable}{S.\arabic{table}}
\setcounter{figure}{0} 
\renewcommand{\thefigure}{S.\arabic{figure}}
\setcounter{footnote}{0} 
\renewcommand{\thefootnote}{\arabic{footnote}}

\makeatletter
\gdef\@extra@binfo{@supp}
\gdef\@extra@b@citeb{@supp}
\makeatother

\begin{bibunit}

\phantomsection
\addcontentsline{toc}{section}{Supplementary Material}

\section*{\begin{center}{Supplementary Material (Not for Publication) for}\end{center}
\begin{center}\textnormal{\LARGE Mechanism Design Without Disclosure:\\Committing to and Running Hidden Mechanisms\\[1em]\large Ran Canetti, Amos Fiat, Yannai A.\ Gonczarowski\\\documentdate}\end{center}}

\section[Proofs of Knowledge of Discrete Log: A Survey and an Extension]{Proofs of Knowledge of Discrete Log:\texorpdfstring{\\}{ }A Survey and an Extension}\label[appendix]{cds}
Our illustrative examples from \cref{examples} use a special form of interactive proofs to prove knowledge of the discrete logarithms of various elements in a large algebraic group $G$. In this \lcnamecref{cds}, we first survey the general definition of this special form of zero-knowledge proofs, called Sigma protocols (not to be confused with our commit-and-run protocols from \cref{framework}), then survey the existing Sigma protocol, called CDS, that we use in Illustrative Examples 1 and 2, and finally present a new Sigma protocol, which extends CDS, that we use in Illustrative Examples 3 and 4.

\subsection{Sigma Protocols}

Sigma protocols are a convenient form of proofs.\cryptofootnote{The formulation of Sigma protocols here restates the standard one without explicitly defining the concept of a security parameter.} Consider a relation $R$ such that $(x,w)\in R$ is interpreted as meaning that $x$ and $w$ together satisfy the property represented by $R$. We are concerned with the problem of proving, for a commonly known $x$, that one knows a value $w$ such that $(x,w)\in R$ without disclosing anything else about $w$. A concrete example of such a relation $R$, used in Illustrative Example 1 from \cref{examples} and surveyed in \cref{schnorr-cds} below, is such that $(x,w)\in R$ for $x=(a_1,\ldots,a_k)$ and $w=(i,\rho)$ if and only if for a prespecified prime $q$ and prespecified $g<q$ it holds that $g^\rho=a_i \bmod q$.

\begin{definition}
A \emph{Sigma protocol} for a relation $R$ is a pair $(P,V)$ of randomized algorithms, where $P$ (standing for ``prover'') is an algorithm with input $(x,w)$, and $V$ (standing for ``verifier'') is an algorithm with input $x$, which are intended to jointly run in the following structured three-message 
interaction: the prover sends a message to the verifier, then the verifier sends a message to the prover called the ``challenge'' message, then the prover sends another message to the verifier called the ``response'' message, and finally, the verifier decides whether to ``accept'' (indicating trust that the prover, whether or not it is in fact the algorithm $P$, knows a value $w'$ such that $(x,w')\in R$). The (single, ``challenge'') message of $V$ 
consists of a uniformly random element from a prespecified domain $\mathcal{C}$, and the final decision made by $V$ is a deterministic function of the common input $x$ and the three messages. Furthermore, the following properties hold:
\begin{itemize}
\item \textbf{Completeness:}
For any $(x,w)\in R$, 
for any first message $\alpha$ sent by $P$ given $(x,w)$, and for any challenge message $\beta\in\mathcal{C}$, we have that $V$ \emph{accepts} when given the \emph{transcript} $(x,\alpha,\beta,\gamma)$, 
 where $\gamma$ is the corresponding response message computed by $P$ after receiving the challenge message $\beta$. In short, denoting by $s$ the private state maintained by $P$ between sending its two messages:
\[
\Pr\bigl[V(x,\alpha,\beta,\gamma)=\text{``accept''}\bigr]=1\text{ for }(\alpha,s)\sim P(x,w), \beta\in\mathcal{C}, \gamma\sim P(x,w,s,\beta).
\]
\item \textbf{Knowledge Extraction:}
There exists a deterministic procedure $K$ whose running time is polylogarithmic in $|\mathcal{C}|$, that given any two $V$-accepted transcripts that share the same input $x$ and first message $\alpha$ but have distinct challenge messages, outputs a witness $w$ such that $(x,w)\in R$. That is, for any $x,\alpha$, and $\beta_1,\beta_2\in\mathcal{C}$ such that $\beta_1\neq\beta_2$, and any $\gamma_1,\gamma_2$ such that $V(x,\alpha,\beta_1,\gamma_1)= V(x,\alpha,\beta_2,\gamma_2)=\textit{accept}$, we have $(x,w)\in R$ where $w=K(x,\alpha,\beta_1,\beta_2,\gamma_1,\gamma_2)$.
\item \textbf{Perfect Witness Indistinguishability:}
For any input $x$ and two witnesses $w_1,w_2$ such that $(x,w_1)\in R$ and $(x,w_2)\in R$, and any (potentially misbehaving, i.e., differing from $V$) verifier algorithm $V'$ that replaces $V$ in interacting with $P$, it holds that the output of $V'$ from interacting with $P$ that has input $(x,w_1)$ is distributed identically to the output of $V'$ from interacting with $P$ that has input $(x,w_2)$.
\end{itemize}
\end{definition}

To understand in what sense Sigma protocols are proofs, and how they relate to our illustrative examples, a short discussion is in order. Consider, e.g., Illustrative Example~1, where the seller proves knowledge of the discrete log, in some base $g$, of one of a sequence of elements $a_1,\ldots,a_k$ of $G$ to the buyer. To do so, we have the seller use a Sigma protocol called CDS (surveyed in \cref{schnorr-cds}), which is a Sigma protocol for the relation $R$ mentioned above, which is defined such that $(x,w)\in R$ for $x=(a_1,\ldots,a_k)$ and $w=(i,\rho)$ if and only if for a prespecified prime $q$ and prespecified $g<q$ it holds that $g^\rho=a_i \bmod q$. Consider the three properties of CDS as a Sigma protocol.

Completeness of CDS guarantees that if the seller knows such $w=(i,\rho)$ and uses $P$ in the interaction, then a buyer who uses $V$ accepts and trusts that the seller indeed knows such $w$.

Knowledge extraction guarantees that after the seller sends the first message $\alpha$---regardless of whether she uses $P$ or not---if she can then respond to a uniformly random challenge with nonnegligible probability of acceptance by $V$, then she in fact knows such $w=(i,\rho)$.

Finally, perfect witness indistinguishability guarantees that if the seller truly knows such $w=(i,\rho)$ and uses $P$ in the interaction, then the buyer (regardless of how cleverly she chooses the challenge message) never learns $i$. To see this, note that because $g$ is a generator of $G$, for every $i'$ there exists $\rho'$ such that $g^{\rho'}=a_{i'}\bmod q$ and hence $(x,(i',\rho'))\in R$. Therefore, perfect witness indistinguishability dictates that the buyer cannot distinguish between the seller using $P$ with input $(x,w)$ and the seller using $P$ with input $(x,(i',\rho'))$. Hence, the buyer cannot learn anything about $i$. (The buyer cannot learn $\rho$ either, because given $\rho$, it is easy to find $i$ by calculating $g^\rho\bmod q$.)

Overall, the three above properties of the Sigma protocols that we survey in \cref{schnorr-cds} and define in \cref{cds-new} guarantee all that we need for our illustrative examples in \cref{examples}. For completeness, in \cref{app:examples} we reformulate a selection of these illustrative examples as commit-and-run protocol catalogs and show that all of our desiderata from \cref{framework} can be (carefully) proven from the three above properties of the Sigma protocols that we use.

\subsection{The Schnorr and CDS Protocols: A Brief Survey}\label[appendix]{schnorr-cds}

We survey the CDS \citep*{CDS94} Sigma protocol for proving knowledge of the discrete logarithm of one out of a set of group elements.

\paragraph{The Schnorr protocol}
As a preliminary step to presenting CDS, we first survey the Schnorr Sigma protocol for proving knowledge of the discrete logarithm of a commonly known group element \citep{S91}.
Let $G$ be an (algebraic) group of large prime order (for the purposes of this paper, let $G$ be as defined in \cref{examples}) and let $g,a$ be elements of $G\setminus\{1\}$. Assume the prover wants to convince the verifier that she knows $\rho$ such that $g^\rho=a$---where here and below all arithmetic operations involving elements of $G$ are within the group (for $G$ as defined in \cref{examples}, this means that all these arithmetic operations are modulo $q$)---while making sure that the verifier learns nothing else about~$\rho$ in the process. 
The protocol (i.e., the interaction between the prover $P_{\text{Schnorr}}$ and verifier $V_{\text{Schnorr}}$ of the protocol) proceeds as follows (where $P_{\text{Schnorr}}$ gets input $(a,\rho)$ such that $g^\rho=a$, and $V_{\text{Schnorr}}$ gets input $a$):
\begin{enumerate}
    \item 
    $P_{\text{Schnorr}}$ draws $r\sim U\{1,\ldots,p\}$, and sends $\alpha=g^r$ (and privately remembers $r$).
    \item
    $V_{\text{Schnorr}}$ draws $\beta\sim U\{1,\ldots,p\}$ and sends it.
    \item
    $P_{\text{Schnorr}}$ sends $\gamma=r+\beta\rho \bmod p$.
    \item
    $V_{\text{Schnorr}}$ accepts if $g^\gamma=\alpha a^{\beta}$.
\end{enumerate}
It is easy to verify completeness (i.e., if $P_{\text{Schnorr}}$ and $V_{\text{Schnorr}}$ interact as above, then $V_{\text{Schnorr}}$ accepts). 

For knowledge extraction,
if $\beta,\gamma$ and $\beta',\gamma'$ are two challenge-response pairs that $V_{\text{Schnorr}}$ accepts with respect to the same first message $\alpha$,
then $g^{\gamma}/a^{\beta}=\alpha=g^{\gamma'}/a^{\beta'}$ and hence $a=g^{\frac{\gamma-\gamma'}{\beta-\beta'}}$, where the calculation in the exponent is modulo $p$.

While perfect witness indistinguishability is moot for proofs of knowledge of discrete logarithms (since the witness is unique), the protocol is \emph{honest-verifier zero knowledge}. That is, the distribution of the transcript $(a,\alpha,\beta,\gamma)$ of the interaction of $V_{\text{Schnorr}}$ and $P_{\text{Schnorr}}$ can be feasibly sampled from given only $a$. Indeed, to draw appropriately distributed $(\alpha,\beta,\gamma)$, first independently draw $\beta,\gamma\sim U\{1,\ldots,p\}$, then set $\alpha=\frac{g^\gamma}{a^\beta}$. We use this property when defining CDS below using the Schnorr protocol.

\paragraph{The CDS protocol}
CDS extends the Schnorr protocol to the case where the prover wants to convince the verifier that she knows the discrete log of \emph{at least one out of several} commonly known elements 
 $a_1,\ldots,a_k\in G\setminus\{1\}$, while making sure that the verifier learns nothing else in the process. In particular, the verifier does not learn to which of $a_1,\ldots,a_k$ the discrete log corresponds. 
 
In a nutshell, the idea is to have the prover $P_{\text{CDS}}$ and the verifier $V_{\text{CDS}}$ of the protocol run $k$ copies of the interaction between $P_{\text{Schnorr}}$ and $V_{\text{Schnorr}}$ from the Schnorr protocol, where $V_{\text{CDS}}$ allows $P_{\text{CDS}}$ to choose the $k$ challenges, subject to the constraint that all the challenges sum up to a single ``master challenge'' that $V_{\text{CDS}}$ chooses. This allows $P_{\text{CDS}}$ to simulate (i.e., sample a transcript for) the Schnorr protocol for all but one of the $k$ copies (the one corresponding to the discrete logarithm that is supposedly known). 
 
The protocol proceeds as follows (where $P_{\text{CDS}}$ gets input $((a_1,\ldots,a_k),(i,\rho))$ such that $g^\rho=a_i$, and $V_{\text{CDS}}$ gets input $(a_1,\ldots,a_k)$):
\begin{enumerate}
    \item $P_{\text{CDS}}$ sends $\alpha_1,\ldots,\alpha_k$, where:
    \begin{enumerate}
    \item $\alpha_i=g^r$, where $r\sim U\{1,\ldots,p\}$
    \item For $i'\neq i$, $\alpha_{i'}=\frac{g^{\gamma_{i'}}}{a_{i'}^{\beta_{i'}}}$, where $\beta_{i'},\gamma_{i'}\sim U\{1,\ldots,p\}$ independently
    \end{enumerate}
    \item
    $V_{\text{CDS}}$ draws $\beta\sim U\{1,\ldots,p\}$ and sends it.
    \item
    $P_{\text{CDS}}$ sets $\beta_i=\beta-\sum_{i'\neq i}\beta_{i'} \bmod p$, $\gamma_i=r+\beta_i\rho \bmod p$, and sends $\gamma_1,\ldots,\gamma_k,\beta_1,\ldots,\beta_k$. 
\item
    $V_{\text{CDS}}$ accepts if $\beta_1+\cdots+\beta_k=\beta \bmod p$, and in addition $g^{\gamma_i}=\alpha_ia_i^{\beta_i}$ for every $i=1,\ldots,k$.
\end{enumerate}

Perfect witness indistinguishability is immediate. Completeness and knowledge extraction are argued similarly to the Schnorr protocol. (For knowledge extraction, if $\beta,(\vec\gamma,\vec\beta)$ and $\beta',(\vec\gamma',\vec\beta')$ are two challenge-response pairs with $\beta\neq \beta'$ that $V_{\text{CDS}}$ accepts with respect to the same first message $\vec\alpha$, then there exists $i$ such that $\beta_i\neq \beta_i'$. For this $i$ we have $g^{\gamma_i}/a_i^{\beta_i}=\alpha_i=g^{\gamma_i'}/a_i^{\beta_i'}$ and hence $a_i=g^{\frac{\gamma_i-\gamma_i'}{\beta_i-\beta_i'}}$, where the calculation in the exponent is modulo $p$.)

\subsection{A New Extension}\label[appendix]{cds-new}
We describe a new extension of CDS\@. The extension considers commonly known values $a_{1,1},\ldots, a_{k,m},g_{1,1},\ldots,g_{k,m}$; in this Sigma protocol the prover proves that there exists $i\in\{1,\ldots,k\}$ such that for every $j$, she knows $\log_{g_{i,j}}a_{i,j}$.

Before we present this extension of CDS, we first present a Sigma protocol that is a straightforward extension of the Schnorr protocol, where a prover proves that for every $j\in\{1,\ldots,m\}$ she knows $\log_{g_j} a_j$, where $a_1,\ldots,a_m$ and $g_1,\ldots,g_m$ are commonly known.
While this could also be done by running $m$ copies of the Schnorr protocol, we avoid doing so because that would not be a useful building block for our extension of CDS.

\paragraph{Extension of Schnorr}
Let $a_1,\ldots,a_m,g_1,\ldots,g_m \in G\setminus\{1\}$. The prover wants to convince the verifier in zero knowledge that she knows $\rho_1,\ldots,\rho_m$ such that $g_j^{\rho_j}=a_j$ for every $j$.
The protocol (i.e., the interaction between the prover $P$ and verifier $V$ of the protocol) proceeds as follows (where $P$ gets input $((a_1,\ldots,a_m,g_1,\ldots,g_m),(\rho_1,\ldots,\rho_m))$ such that $g_j^{\rho_j}=a_j$ for every $j$, and $V$ gets input $(a_1,\ldots,a_m,g_1,\ldots,g_m)$):
\begin{enumerate}
    \item 
    $P$ independently draws $r_1,\ldots,r_m\sim U\{1,\ldots,p\}$, and sends $\vec\alpha=\alpha_1,\ldots,\alpha_m$, where $\alpha_j=g_j^{r_j}$ for every $j$.
    \item
    $V$ draws $\beta\sim U\{1,\ldots,p\}$ and sends it.
    \item
    $P$ sends $\gamma_1,\ldots,\gamma_m$, where $\gamma_j=r_j+\beta \rho_j \bmod p$ for every $j$.
    \item
    $V$ accepts if $g_j^{\gamma_j}=\alpha_j a_j^{\beta}$ for every $j=1,\ldots,m$.
\end{enumerate}
Completeness, knowledge extraction, and perfect witness indistinguishability (and honest-verifier zero knowledge) follow in the same way as for the Schnorr protocol. We emphasize that the proof uses only a single verifier challenge (the same challenge for all $j$); this property will be useful for the extension of CDS described next.

\paragraph{Extension of CDS}
Let $a_{i,j},g_{i,j}\in G\setminus\{1\}$ for every $i=1,\ldots,k$ and $j=1,\ldots,m$. The prover wants to convince the verifier in zero knowledge that she knows $(i,\rho_1,\ldots,\rho_m)$ such that $g_{i,j}^{\rho_j}=a_{i,j}$ for every $j$. The protocol (i.e., the interaction between the prover $P^*$ and verifier $V^*$ of the protocol) proceeds as follows (where $P^*$ gets input $((a_{1,1},\ldots,a_{k,m},g_{1,1},\ldots,g_{k,m}),\linebreak(i,\rho_1,\ldots,\rho_m))$ such that $g_{i,j}^{\rho_j}=a_{i,j}$ for every $j$, and $V^*$ gets input $(a_{1,1},\ldots,a_{k,m},g_{1,1},\ldots,g_{k,m})$):
\begin{enumerate}
    \item $P^*$ draws $\beta_{i'}\sim U\{1,\ldots,p\}$ independently for each $i'\ne i$ and sends $\alpha_{1,1},\ldots,\alpha_{k,m}$, where, for every $j=1,\ldots,m$:
    \begin{enumerate}
    \item $\alpha_{i,j}=g_{i,j}^{r_j}$, where $r_j\sim U\{1,\ldots,p\}$ independently
    \item For $i'\neq i$, $\alpha_{i',j}=\frac{g_{i',j}^{\gamma_{i',j}}}{a_{i',j}^{\beta_{i'}}}$, where $\gamma_{i',j}\sim U\{1,\ldots,p\}$ independently
    \end{enumerate} 
    \item
    $V^*$ draws $\beta\sim U\{1,\ldots,p\}$ and sends it.
    \item
    $P^*$ sets $\beta_{i}=\beta-\sum_{i'\neq i}\beta_{i'} \bmod p$, and $\gamma_{i,j}=r_j+\beta_{i} \rho_j \bmod p$ for every $j$, and sends $\gamma_{1,1},\ldots,\gamma_{k,m},\beta_{1},\ldots,\beta_{k}$.
    \item
    $V^*$ accepts if $\beta_{1}+\cdots+\beta_{k}=\beta \bmod p$, and in addition $g_{i,j}^{\gamma_{i,j}}=\alpha_{i,j} a_{i,j}^{\beta_{i}}$ for every $i=1,\ldots,k$ and $j=1,\ldots,m$.
\end{enumerate}

Completeness and perfect witness indistinguishability are argued similarly to CDS (replacing the Schnorr protocol with its above extension as a building block). For knowledge extraction, if $\beta,(\vec\gamma,\vec\beta)$ and $\beta',(\vec\gamma',\vec\beta')$ are two challenge-response pairs with $\beta\neq \beta'$ that $V^*$ accepts with respect to the same first message $\vec\alpha$, then there exists $i$ such that $\beta_i\neq \beta_i'$. 
For this $i$ and for every $j=1,\ldots,m$ we have $g_{i,j}^{\gamma_{i,j}}/a_{i,j}^{\beta_i}=\alpha_{i,j}=g_{i,j}^{\gamma_{i,j}'}/a_{i,j}^{\beta_i'}$ and hence $a_{i,j}=g_{i,j}^{\frac{\gamma_{i,j}-\gamma_{i,j}'}{\beta_i-\beta_i'}}$, where the calculation in the exponent is modulo $p$.

\section[Contract Design Without Disclosure:\texorpdfstring{\\}{ }Private Actions and Moral Hazard]{Contract Design Without Disclosure:\texorpdfstring{\\}{ }Private Actions and Moral Hazard}\label[appendix]{contracts}

In this paper, we have considered mechanisms with private types and public actions (essentially, reports). It is straightforward to apply our framework also to mechanisms with private actions (and hence moral hazard), such as contracts.\footnote{In a nutshell: A principal wants an agent to invest effort to stochastically increase the principal's returns. The level of effort that the agent exerts is private (unobserved and noncontractible), so the principal offers the agent a \emph{contract}: a monetary transfer between the principal and the agent that is a function of the principal's future (stochastic, publicly realized) returns. As we discuss in the introduction, there are natural settings in which it can be of interest for the principal to commit to this function while hiding it and proving that it incentivizes exertion of effort, while having the agent eventually learn only her bottom-line realized wage (or loss, in cases without limited liability).} In a traditional protocol for contracts, the properties to be proven by the mechanism designer (the principal) in the commitment step are individual rationality (and in some settings limited liability) and the optimal level of effort for the player (the agent) to privately exert; furthermore, the analogue of the direct revelation step from auction protocols is an actions/returns step in which the agent privately exerts effort, following which the stochastic returns are publicly realized. In the run step, the principal then calculates the wage based on the returns.

A commit-and-run protocol for such contracts is then similar to the one in \cref{protocols}, with the same changes as just described to the incentive property proven and to the direct revelation-turned-actions/returns step, and---denoting the returns, by abuse of notation, by $t$---with no changes to the remainder of the protocol. At a high level, the commitment playbook then involves the principal sending to the agent a cryptographic commitment to a hidden contract along with a zero-knowledge proof of the optimal effort for the agent to exert, and the run playbook again involves a zero-knowledge proof that the declared outcome is consistent with the hidden contract/mechanism that was committed to in the commitment step.

\section{The Illustrative Examples from Section~\ref{examples} as\texorpdfstring{\\}{ }Commit-and-Run Protocols}\label[appendix]{app:examples}
In this \lcnamecref{app:examples}, for completeness, we reformulate Illustrative Examples 1 and 2 from \cref{examples} as commit-and-run protocol catalogs (up to the added interaction during proofs) and prove that they satisfy our desiderata from \cref{framework}. The analysis for Illustrative Examples 3 and 4 is similar. While we do not do so for simplicity of our constructions, we note that under certain computational infeasibility assumptions, all verifier interaction can be removed from CDS as well as from our new extension of it (i.e., the verifier's message can be eliminated from the interaction without hurting any of the properties\cryptofootnote{By this we mean that the Sigma protocols that we use can be made into non-interactive zero-knowledge proofs of knowledge (in the random oracle model), using the Fiat--Shamir transform \citep{FS86}.}), resulting in commit-and-run protocol catalogs without any added interaction during proofs.

We formalize Illustrative Example 2 from \cref{examples}; Illustrative Example 1 is a special case.
Let $H\in\mathbb{N}$ be a power of $2$. Let $T=T_1=\{0,\ldots,H\!-\!1\}^2$, each $t\in T$ indicating the buyer's value for each item. We choose a mechanism-description language in which each mechanism description is exactly $2\log_2H$ bits long, interpreted as pairs of prices in $\{0,\ldots,H\!-\!1\}$. It suffices to only consider specifications with $L_a=2\log_2H$ (bigger $L_a$ would be equivalent, and smaller $L_a$ would have no valid mechanism descriptions), with $L_p=0$ (since all mechanisms describable in our language are IR and IC), and with $\runrandlen=0$ (since all such mechanisms are deterministic). Running any such mechanism on a given type report can always be done in time linear in $L_a$, say, in time $100\cdot L_a$, so it suffices to only consider specifications with $R=100\cdot L_a$. So, we must construct a protocol for every $B\in\mathbb{N}$ and $0<\varepsilon<1$.

Let $G$ be the group from \cref{examples}, for $p$ and $q$ to be chosen later as a function of $B$ and~$\varepsilon$. For ease of presentation, we assume that $\commitrand$, rather than being a uniformly random sequence of $\commitrandlen$ bits, is a pair $(g,h)$ of uniformly random elements in $G\setminus\{1\}$. We take $\commitmsglen=2\log_2H\lceil\log_2 q\rceil$ and interpret a commitment message as $2\log_2H$ commitments to bits ($\log_2H$ commitments for each price, each of them an element of $G\setminus\{1\}$). $\commitverify$ simply verifies that each of these $2\log_2H$ bit commitments is an element of $G\setminus\{1\}$.\footnote{This can be feasibly done (i.e., done in time polylogarithmic in~$|G|$) for the groups based on Sophie Germain primes and safe primes as mentioned in \cref{examples}, since verifying that an element of the larger group is a square (i.e., is indeed an element of the smaller group) can be feasibly done since the order of the big group, $q-1$, is known.}

As already noted, the proof of properly running the mechanism is interactive in this protocol, so instead of defining $\runmsglen$ and $\runverify$, we define an interactive verification process. We will be mindful that the total computational power required by the buyer during the run phase (i.e., after the types are reported) satisfies the computational constraints we place on $\runverify$ in computationally unobtrusive protocol catalogs.

The playbook is defined as in \cref{examples}: $\protcommitstrat$ commits to two prices as described there, and $\protrunstrat$ (interactively) proves correctness as described there, and uses (the interactive) CDS (see \cref{cds}) if a zero-knowledge proof of knowledge of a discrete logarithm in a given base of one or more elements from a set is called for.

The fact that our protocol is \emph{implementing} follows from the completeness of CDS (see \cref{cds}). To see that it is \emph{hiding}, observe that the commitment step is hiding since each bit commitment, regardless of the bit that it is hiding, has precisely the same distribution: uniform in $G\setminus\{1\}$, and hence conveys no information whatsoever regarding the hidden bit. The run step is hiding due to the perfect witness indistinguishability of CDS (see \cref{cds}). Note that the same protocol is hiding for all $B$ and $\varepsilon$: its hiding guarantee is absolute---even a computationally unbounded adversary (a stronger guarantee than for a computationally bounded one) cannot guess the origin of a transcript with probability greater than half (a stronger guarantee than not guessing with probability nonnegligibly greater than half).

It remains to show that a group size $p$ can be chosen based on $B$ and $\varepsilon$ so that the protocol can be shown to be both \emph{committing} and \emph{computationally unobtrusive}.
Let $B'=5\cdot|T|\cdot B/\varepsilon$. There is broad confidence among cryptographers and computational complexity theorists that choosing the logarithm, $\log p$, of the size of the group $G$ to be a sufficiently large polynomial in $\log(B'/\varepsilon)$ suffices for guaranteeing that no algorithm with computing power at most $B'$ can calculate the discrete logarithm of $h$ base $g$ for uniformly random $g,h$ in $G\setminus\{1\}$ with probability greater than~$\nicefrac{\varepsilon}{6}$. This widely believed computational infeasibility assumption is used in practice in many real-life cryptographic systems, and breaking it would result in the ability to break into those systems. Note that with this choice of $p$, representing elements of $G$, drawing random elements of $G$, calculating multiplication and exponentiation in $G$, and arithmetic (of exponents) $\bmod p$ can all be done in time polylogarithmic (and, in particular, subpolynomial) in $B$ and $1/\varepsilon$, and therefore the protocol is \emph{computationally unobtrusive}. To prove that the protocol is \emph{committing} with this choice of $p$ (given the above widely believed computational infeasibility assumption), we show that an adversary that breaks the property just defined can be used to build an algorithm with computing power at most $B'$ that can calculate the discrete logarithm of $h$ base $g$ for uniformly random $g,h$ in~$G\setminus\{1\}$ with probability greater than~$\nicefrac{\varepsilon}{6}$, contradicting the computational infeasibility assumption.

Let $(\commitstrat,\runstrat)$ be a mechanism-designer strategy computable in time at most $B$. We implicitly define an appropriate function $\alpha:\bigl(G\setminus\{1\}\bigr)^2\rightarrow\Delta\bigl(\{0,\ldots,H\!-\!1\}^2\bigr)$ by directly defining the coupling $\Xi$ of $\commitstrat$ and $\alpha$. (We give an algorithmic definition because we will later be interested in the running time of $\alpha(\commitrand)$.)
\begin{enumerate}
    \item Let $(\commitmsg,s)\sim \commitstrat(\commitrand)$.
    \item If not $\commitverify(\commitrand,\commitmsg)$, then return $\bigl((\commitmsg,s),(s^1,s^2)\bigr)$ for arbitrary $(s^1,s^2)\in\{0,\ldots,H\!-\!1\}^2$.
    \item Initialize $s^1,s^2$ to be the maximum possible price, $H\!-\!1$.
    \item For each $t=(v_1,v_2)\in T$:
    \begin{enumerate}
        \item Let $(x,\runmsg)\sim\runstrat\bigl(\commitrand,(\commitmsg,s),(v_1,v_2)\bigr)$. (Recall that $\runrandlen=0$.)
        \item If the first (respectively, second) item is sold in $x$ and the verification of the revealed price passes (i.e., the commitment to the price is opened to the price claimed in the outcome), then update $s^1$ (respectively, $s^2$) to be the price at which it is sold.
    \end{enumerate}
    \item Return $\bigl((\commitmsg,s),(s^1,s^2)\bigr)$.
\end{enumerate}

We now prove that $\alpha$ (under the coupling $\Xi$) satisfies the guarantee in the definition of a committing protocol. Assume for contradiction that for some $(\commitstrat,\runstrat)$ computable in time at most $B$ there exists $t=(v_1,v_2)\in T$ such that with probability higher than $\varepsilon$ verification passes but the outcome is not that of $M_{\alpha(\commitrand)}$. We use $(\commitstrat,\runstrat)$ to construct an algorithm $\Pi$ with running time at most $B'$ (as defined above) that takes two random elements $g,h\in G\setminus\{1\}$ and with probability greater than $\nicefrac{\varepsilon}{6}$ computes $\ell$ such that $g^\ell=h$ (in $G$). This contradicts the assumption of the infeasibility of calculating a discrete logarithm in $G$.

The algorithm $\Pi$ runs as follows on input $(g,h)$. Set $\commitrand=(g,h)$ and run $\Xi(\commitrand)$ to obtain a commitment message $\commitmsg=(C^1_1,\ldots,C^1_{\log_2H},C^2_1,\ldots,C^2_{\log_2H})$ and state $s$ (jointly distributed as if output by $\commitstrat$) as well as a mechanism description $A=(s^1,s^2)$ (distributed as if output by $\alpha$). 
Next, run $\runstrat$ with commitment bits $\commitrand$, commitment message $\commitmsg$ and state $s$, and type report $(v_1,v_2)$ to obtain the outcome $x$. If the outcome declared by $\runstrat$ is \emph{not} $M_{(s^1,s^2)}(v_1,v_2)$, then proceed as follows. We consider two cases. The first case is where the outcome $x$ sells one of the items---without loss of generality, item $1$---for a price different than the one prescribed by $(s^1,s^2)$. In this case, verification is by opening the commitment to the price of item $1$, and therefore verification does not involve any challenges. Finish running $\runstrat$ to run the verification. If verification passes, then it means that there exists $t'=(v_1',v_2')\in T$ (the type that defines the price $s^1$ during the run of~$\Xi$) such that when run with $t$ and $t'$, the run strategy $\runstrat$ opened the commitment to the price of item $1$ into two different prices. In turn, this means that $C^1_j$ for some $j$ was opened by $\runstrat$ into both $0$ and $1$, i.e., it outputs values $\rho,\rho'$ such that $g^\rho=h^{\rho'}=C^1_j$. In this case (if verification passes), $\Pi$ computes $\ell=\rho\cdot{\rho'}^{-1}\bmod p$ (this can be feasibly done) and returns it.

The second case is where the outcome $x$ does not sell an item---without loss of generality, item $1$---even though based on the prices $(s^1,s^2)$ it should have been sold. In this case, verification is by CDS and involves a random challenge. This means that there exists $t'=(v_1',v_2')\in T$ (again, the type that defines the price $s^1$ during the run of $\Xi$) such that when running with $(v_1',v_2')$, $\runstrat$ opens the commitment to the price of item $1$ into the price $s^1$ while when running with $(v_1,v_2)$, $\runstrat$ claims that the price of item $1$ is greater than~$s'$ for some value $s'\ge s^1$, and attempts to prove this during verification. More specifically, there exist $i_1,\ldots,i_k$ such that when running with $(v_1',v_2')$ every $C^1_{i_j}$ is opened to $0$, i.e., a discrete log base $g$ of $C^1_{i_j}$ is revealed, but when running with $(v_1,v_2)$, $\runstrat$ claims that it knows how to represent one of these commitments as a commitment to $1$, i.e., knows a discrete log base $h$ of one of the $C^1_{i_j}$s, and attempts to prove this using CDS during verification.
In this case, $\Pi$ repeats the following $\lceil4/\varepsilon\rceil$ times: sample $\beta\sim U\{1,\ldots,p\}$ and finish running $\runstrat$ from immediately after the first message of CDS was sent, to continue the verification with challenge $\beta$. If verification passes for at least two sampled values $\beta^1\ne\beta^2$, then run the knowledge extractor of CDS (see \cref{cds}) on the verification transcripts for these two values to obtain a discrete logarithm base $h$ of one of the $C^1_{i_j}$s. We then once again have discrete logarithms in both base $g$ and base $h$ of a commitment, so we continue in the same way as in the first case.

Note that the running time of $\Pi$ is indeed at most $B'$ (it runs $\commitstrat$ and/or $\runstrat$ at most $|T|\!+\!\lceil4/\varepsilon\rceil$ times, and in addition only performs operations that are polylogarithmic in $|G|$, which are fast enough to not take the overall running time over $B'$). It remains to show that $\Pi$ succeeds with probability greater than $\nicefrac{\varepsilon}{6}$ over uniformly chosen $g,h$. Note that a probability greater than $\varepsilon$ that verifications pass but the outcome is not $M_{\alpha(\commitrand)}(t)$ implies that with probability at least $\nicefrac{\varepsilon}{2}$ over the randomness of $\commitrand$, $\commitstrat$, and $\runstrat$ up to and including sending the first CDS message, there is probability at least $\nicefrac{\varepsilon}{2}$ over drawing the CDS challenge and the remaining randomness of $\runstrat$ that verifications pass but the outcome is not $M_{\alpha(\commitrand)}(t)$. (Otherwise, the probability that verifications pass but the outcome is not $M_{\alpha(\commitrand)}(t)$ is less than $\nicefrac{\varepsilon}{2}\cdot1+(1-\nicefrac{\varepsilon}{2})\cdot\nicefrac{\varepsilon}{2}<\varepsilon$.) Whenever the former $\nicefrac{\varepsilon}{2}$-probability event holds, we have that $\Pi$ succeeds with probability at least $\bigl(1-(1-\nicefrac{\varepsilon}{2})^{2/\varepsilon}\bigr)^2>(1-\nicefrac{1}{e})^2>\nicefrac{1}{3}$ (the first expression is the probability of at least one run-verification passing in the first $\nicefrac{2}{\varepsilon}$ samples and at least one run-verification passing in the last $\nicefrac{2}{\varepsilon}$ samples; the inequalities are loose enough to absorb the small probability that $\beta^1=\beta^2$). Therefore, $\Pi$ succeeds with probability greater than $\nicefrac{\varepsilon}{6}$ over uniformly random $(g,h)\in (G\setminus\{1\})^2$, as required.\qed

\section[Mapping the Commit-then-Prove Scheme to the\texorpdfstring{\\}{ }Cryptographic Literature]{Mapping the Commit-then-Prove Scheme to the Cryptographic Literature}\label[appendix]{cp-proof-sec}

In this \lcnamecref{cp-proof-sec}, we provide an overview of the main components from the cryptographic literature that jointly substantiate the construction and properties of the commit-then-prove scheme guaranteed by \cref{ctp} from \cref{cp-sec}. To keep the number of pages manageable, this \lcnamecref{cp-proof-sec} (and only this \lcnamecref{cp-proof-sec}) is written for an audience with background in cryptography.

Recall that a commit-then-prove scheme $\pi$ for some publicly known relation~$\ctprelation$ allows a committer to commit to a secret string $w$, and later provide, over time, public statements $\phi_1,\phi_2,\ldots$, where each $\phi_i$ is accompanied by a proof that the committer knows some secret supporting information $P_i$ such that 
$\ctprelation(w,\phi_i,P_i)$ holds. This is done while preserving the secrecy of $w$ and of the $P_i$s, in the sense that interacting with the committer provides the verifiers with no additional information or computational advantage beyond what can be inferred from the statements $\phi_1,\phi_2,\ldots$ themselves. 
Furthermore recall that the commit stage, and each prove stage, consists of a single message generated by the committer.

More specifically, $\pi$ is a commit-then-prove scheme for relation $\ctprelation$ if $\pi$ UC-realizes $\fctp^\ctprelation$, with the following additional properties. (1) Perfect completeness: when the committer follows the protocol, the verifier is guaranteed to always accept. (2) The only common state of the committer and the verifier programs is a uniformly distributed reference string. (3) The simulated reference string is also distributed uniformly. (This provision is later used to guarantee our analogue of soundness---i.e., what we call \emph{committing}---in the commit-and-run protocols for each value of the reference string, except perhaps for a negligible fraction.)\footnote{The UC framework allows the simulator to generate a reference string that is distributed differently than the real reference string, as long as this fact remains unnoticed by the adversary. In contrast, the $\fctp$-challenge of Definition~\ref{def:ctp} is a direct restatement of the requirement that $\fctp$ is UC-realized with the aid of a common uniform string---with the exception that Definition~\ref{def:ctp} restricts attention to simulators that generate a reference string that is uniformly distributed.
}

We also note that $\fctp$ is similar to the Commit-and-Prove functionality from \cite{clos02}, with the following additional exceptions:
\begin{enumerate}
\item
In \cite{clos02}, the committer can commit, over time, to multiple secret values $w_1,w_2,\ldots$, and then prove that each public statement $\phi_i$ satisfies $\ctprelation(\phi_i,w_1,w_2,\ldots)$ with respect to multiple committed values. This also obviates the need to have explicit supporting information $P_1,P_2,\ldots$\,.
\item
In \cite{clos02} each commitment and proof stage is allowed to be interactive, while $\fctp$ requires non-interactivity.
\end{enumerate}

Constructions of the commit-then-prove scheme in this work will thus follow a similar structure to those in \cite{clos02}, while replacing components so as to satisfy the additional requirements. Recall that in \cite{clos02} the commitment stage can be realized via any protocol that UC-realizes $\fcom$, the ideal commitment functionality, whereas each proof stage can be realized via any protocol that UC-realizes $\fzk$, the zero-knowledge proof functionality. In particular, the UC commitment protocols and zero-knowledge proofs from \cite{CF01} work.

In the case of the commit-then-prove scheme, we use a similar strategy, except that here we replace $\fcom$ and $\fzk$ by their non-interactive variants $\fnicom$ and $\fnizk$ \citep[see, e.g.,][]{CSW22}, which provide the additional guarantee that the commit and prove stages are each realized via a single message.\footnote{The proof that the combination of a protocol that realizes $\fnicom$ and a protocol that realizes $\fnizk$ in fact realizes $\fctp$ is the same as in \cite{clos02}. In fact, the non-interactivity guarantees of $\fnicom$ and $\fnizk$ allow simplifying the proof considerably. We omit further details.} 

It thus remains to show how to appropriately realize $\fnicom$ and $\fnizk$. We first show protocols that realize $\fnicom$ and $\fnizk$ with a common uniformly random string. As demonstrated in \cite{CSW22}, the \cite{CF01} commitment protocol UC-realizes $\fnicom$, provided a public-key encryption scheme that is secure against chosen-ciphertext attacks and a claw-free pair of permutations. The reference string is then interpreted as a pair $(e,c)$ where $e$ is a public key for the encryption scheme and $c$ is a key of the claw-free pair. Both primitives have instantiations based on standard, widely believed computational infeasibility assumptions and where the relevant keys are uniformly drawn. The protocol of \cite{GOS12} UC-realizes $\fnizk$ with a uniformly chosen reference string, under similar computational infeasibility assumptions. We further note that: (1) the simulators in both \cite{CF01} and \cite{GOS12} indeed use a reference string that is distributed identically to the one used in the respective protocols, and (2) both protocols enjoy perfect completeness: as long as the parties follow the protocol, the verifier always accepts. (We stress that these instantiations are only intended as proof-of-concept constructions; multiple other instantiations exist, offering a rich ground for trading off various complexity parameters and computational infeasibility assumptions.)

In terms of concrete infeasibility assumptions, we recall that non-interactive UC commitments and zero knowledge can be realized from any number of computational infeasibility assumptions that are standard and widely believed in cryptography, for instance the RSA assumption, the LWE (learning with errors) assumption, the DDH (decisional Diffie--Hellman) assumption together with the LPN (learning parity with noise) assumption, or the bilinear DDH assumption \citep[see, e.g.,][]{CF01,clos02,GOS12,BKM20,CSW22}.

\small
\singlespacing
\renewcommand\refname{References for Supplementary Material}
\phantomsection
\addcontentsline{toc}{section}{References for Supplementary Material}
\putbib
\end{bibunit}

\end{document}